%% file: WidebandHP_SC_Rev_v1_T4.tex
\documentclass[onecolumn,draftclsnofoot,12pt]{IEEEtran}
\input{input.tex}

\newcommand{\pinv}[1]{\ensuremath{#1^{\dagger}}} 	% Moore-Penrose pseudo-inverse

\newcommand{\sref}[1]{{Section}~\ref{#1}}
\def\j{\mathrm{j}}
\usepackage{epstopdf}
\usepackage{varwidth}
\usepackage{enumerate}
\usepackage{algorithmicx}
\usepackage{algorithm}
\usepackage{amsmath}
\usepackage[noend]{algpseudocode}
\usepackage{float}
\usepackage{color}
\usepackage{makeidx}
\usepackage{bbm}
\usepackage{makecell}

\allowdisplaybreaks

\begin{document}
\title{ Frequency Selective Hybrid Precoding for Limited Feedback Millimeter Wave Systems}
\author{Ahmed Alkhateeb and Robert W. Heath, Jr. \\ (Invited Paper)\thanks{Ahmed Alkhateeb and Robert W. Heath Jr. are with The University of Texas at Austin (Email: aalkhateeb, rheath@utexas.edu).} \thanks{This work is supported in part by the National Science Foundation under Grant No. 1319556, and by a gift from Nokia.}}
\maketitle

\begin{abstract}
Hybrid analog/digital precoding offers a compromise between hardware complexity and system performance in millimeter wave (mmWave) systems. This type of precoding allows mmWave systems to leverage large antenna array gains that are necessary for sufficient link margin, while permitting low cost and power consumption hardware. Most prior work has focused on hybrid precoding for narrow-band mmWave systems, with perfect or estimated channel knowledge at the transmitter. MmWave systems, however, will likely operate on wideband channels with frequency selectivity. Therefore, this paper considers wideband mmWave systems with a limited feedback channel between the transmitter and receiver. First, the optimal hybrid precoding design for a given RF codebook is derived. This provides a benchmark for any other heuristic algorithm and gives useful insights into codebook designs. Second, efficient hybrid analog/digital codebooks are developed for spatial multiplexing in wideband mmWave systems. Finally, a low-complexity yet near-optimal greedy frequency selective hybrid precoding algorithm is proposed based on Gram-Schmidt orthogonalization. Simulation results show that the developed hybrid codebooks and precoder designs achieve very good performance compared with the unconstrained solutions while requiring much less complexity.
\end{abstract}

%%%%%%%%%%%%%%%%%%%%%%%%%%%%%%%%%%%%%%%%%%%%%%%%%%%%%%%%%%%%%%%%%%%%%%%%%%%%%%%%%%%%%%%%%%%%%%%%%%%%%%%%%%%%%
\section{Introduction} \label{sec:Intro}
%%%%%%%%%%%%%%%%%%%%%%%%%%%%%%%%%%%%%%%%%%%%%%%%%%%%%%%%%%%%%%%%%%%%%%%%%%%%%%%%%%%%%%%%%%%%%%%%%%%%%%%%%%%%%

Millimeter wave (mmWave) communication can leverage the large bandwidths potentially available at millimeter wave carrier frequencies to provide high data rates \cite{Rappaport2014}. This makes mmWave a promising carrier frequency for 5G cellular systems \cite{Pi2011,Boccardi2014,Bai2014,Andrews2014,Rangan2014,Wang2015}. Recent channel measurements have confirmed the feasibility of using mmWave not only for backhaul \cite{Wang2015,MacCartney2014,Hur2013}, but also for the access link \cite{Rappaport2013a}. Further, system level evaluation of mmWave network performance indicate that mmWave cellular systems can achieve a similar spectral efficiency to that obtained at lower-frequency while providing orders of magnitudes more data rate thanks to the larger bandwidth \cite{Bai2015,Akdeniz2014,Singh2015,Zhu2014}. Though mmWave cellular is recently of interest for 5G, it was proposed as early as thirty years ago \cite{Walke1985}. MmWave wireless communication has been considered for many other applications beyond cellular systems including wireless local area networks \cite{11ad}, personal area networks \cite{WirelessHDSrandard2007}, wearable device communications \cite{Pyattaev2015,Venugopal2015}, joint vehicular communication and radar systems \cite{Wehling2005,Heddebaut2010,Kumari2015}, and simultaneous energy/data transfer \cite{Khan2015,Krikidis2014,Bi2015}.

To guarantee  sufficient received signal power  at  mmWave frequencies, large antenna arrays are beneficial at both the transmitter and receiver \cite{Pi2011,Rappaport2014,Rappaport2013a,ElAyach2014,Han2015,Roh2014}. Fortunately, large antenna arrays can be packed into small form factors due to the small mmWave antenna size \cite{Biglarbegian2011,Sun2013}. Exploiting large arrays using multiple input multiple output (MIMO) signal processing techniques like precoding and combining, however, is different at mmWave compared with sub-6 GHz solutions. This is mainly due to the different hardware constraints on the mixed signal components because of their high cost and power consumption \cite{Alkhateeb2014d}. Further, the best precoders are designed based on instantaneous channel state information, which is difficult to acquire at the transmitter in large mmWave systems \cite{Alkhateeb2014d} due to the high channel dimensionality. Therefore, developing precoding algorithms and codebooks for limited feedback wideband mmWave systems is important for building these systems.

 \subsection{Prior Work}
 Precoding and combining is a classic topic in MIMO communications. The use of channel-dependent precoding at the transmitter is a result of the derivation of the MIMO channel capacity \cite{Teletar1999}. Initial work was focused on deriving optimum precoders and combiners under different criteria \cite{Yang1994,Scaglione2002,Palomar2003}. As the importance of channel state information was realized, effort shifted to the development of limited feedback precoding techniques, where the precoder is selected from a codebook of possible precoders known to both the transmitter and receiver \cite{Love2008,Love2005}. Limited feedback precoding is now an important part of commercial wireless communication systems including LTE \cite{Lee2009}, IEEE 802.11n \cite{IEEE11n2012}, and IEEE 802.11ac among others \cite{Bejarano2013}, etc. 
  
For the sake of low power consumption in consumer-based wireless systems, beamforming at mmWave has been primarily realized in analog, using networks of phase shifters in the RF domain \cite{Wang2009,Xia2008b}. This reduces the number of required RF chains, and consequently the cost and power consumption. The analog-only beamforming solution is also supported in commercial indoor mmWave communication standards like IEEE 802.11ad \cite{11ad} and wirelessHD \cite{WirelessHDSrandard2007}. For MIMO-OFDM systems, analog-only post-IFFT and pre-FFT beamforming was proposed for different criteria such as capacity and SNR maximization \cite{Via2010}. Analog beamforming as in \cite{Wang2009,Xia2008b,11ad,Via2010}, though, is limited to single-stream transmission. Further, analog beamformers are subject to additional hardware constraints. For example, the phase shifters might be digitally controlled and have only quantized phase values and adaptive gain control might not be implemented. This limits the ability to make sophisticated processing in analog-only solutions.

Hybrid analog/digital precoding, which divides the precoding between analog and digital domains, was proposed to handle the trade-off between the low-complexity limited-performance analog-only solutions and the high-complexity good-performance fully digital precoding \cite{Zhang2005a,Venkateswaran2010,ElAyach2014,Alkhateeb2013,Alkhateeb2014,Sohrabi2015,Mendez-Rial2015a,Chen2015,Kim2013}. The main advantage of hybrid precoding over conventional precoding is that it can deal with having fewer RF chains than antennas. For general MIMO systems, hybrid precoding for diversity and  multiplexing gain were investigated in \cite{Zhang2005a}, and for interference management in \cite{Venkateswaran2010}.  These solutions, however, did not make use of the special mmWave channel characteristics in the design as they were not specifically developed for mmWave systems. In \cite{ElAyach2014}, the sparse nature of mmWave channels was exploited; low-complexity iterative algorithms based on orthogonal matching pursuit were devised, assuming perfect channel knowledge at the transmitter. Extensions to the case when only partial channel knowledge is required and when the channel and hybrid precoders are jointly designed were considered in \cite{Alkhateeb2013,Alkhateeb2014}. Algorithms that do not rely on orthogonal matching pursuit were proposed in \cite{Sohrabi2015,Mendez-Rial2015a,Chen2015} for the hybrid precoding design with perfect channel knowledge at the transmitter. The main objective of these algorithms was to achieve an achievable rate that approaches the rate achieved by fully-digital solutions. The work in \cite{ElAyach2014,Alkhateeb2013,Alkhateeb2014,Sohrabi2015,Mendez-Rial2015a}, though, assumed a narrow-band mmWave channel, with perfect or partial channel knowledge at the transmitter. In \cite{Kim2013}, hybrid beamforming with only a single-stream transmission over MIMO-OFDM system was considered. The solution in \cite{Kim2013} though relied on the joint exhaustive search over both RF and baseband codebooks without giving specific criteria for the design of these codebooks. As mmWave communication is expected to employ broadband channels, developing spatial multiplexing hybrid precoding algorithms for wideband mmWave systems is important. Further, acquiring the large mmWave MIMO channels at the transmitter is difficult, which highlights the need to devise limited feedback hybrid precoding solutions.

\subsection{Contribution}

In this paper, we develop hybrid precoding solutions  and codebooks for limited feedback wideband mmWave systems.  In our proposed system, the digital precoding is done in the frequency domain and can be different for each subcarrier, while the RF precoder is frequency flat. The contributions of this paper are summarized as follows.
\begin{itemize} 
	\item{First, we consider a frequency-selective hybrid precoding system with the RF precoders taken from a quantized codebook. For this system, we derive the optimal hybrid precoding design that maximizes the achievable mutual information  under total power and unitary power constraints. Even though an exhaustive search over the RF codebook will be still required, the derived solution provides insights into hybrid analog/digital codebooks and greedy hybrid precoding design problems. Further, this solution gives a benchmark for the other heuristic algorithms that can be useful for evaluating their performance.}	
	\item{Second, we consider a limited feedback frequency-selective hybrid precoding system where both the baseband and RF precoders are taken from quantized codebooks. For this system, we develop efficient hybrid analog and digital precoding codebooks that attempt to minimize a distortion function defined by the average mutual information loss due to the  quantized hybrid precoders when compared with the unconstrained digital solution.}
	\item{Finally, we design a greedy hybrid precoding algorithm based on Gram-Schmidt orthogonalization for limited feedback frequency selective mmWave systems. Despite its low-complexity, the proposed algorithm is illustrated to achieve a similar performance compared with the optimal hybrid precoding design that  requires an exhaustive search over the RF and baseband codebooks.}
\end{itemize}
The performance of the proposed codebooks and precoding algorithms is evaluated by numerical simulations in wideband mmWave setups, and compared with digital only precoding schemes in \sref{sec:Results}. 

% Notation
We use the following notation throughout this paper: $\bA$ is a matrix, $\ba$ is a vector, $a$ is a scalar, and $\cA$ is a set. $|a|$ and $\measuredangle{a}$ are the magnitude and phase of the complex number $a$. $|\bA|$ is the determinant of $\bA$, $\|\bA \|_F$ is its Frobenius norm, whereas $\bA^T$, $\bA^*$, $\bA^{-1}$, $\pinv{\bA}$ are its transpose, Hermitian (conjugate transpose), inverse, and pseudo-inverse respectively. $[\bA]_{r,:}$ and $[\bA]_{:,c}$ are the $r$th row and $c$th column of the matrix $\bA$, respectively. $\mathrm{diag}(\ba)$ is a diagonal matrix with the entries of $\ba$ on its diagonal. $\bI$ is the identity matrix and $\mathbf{1}_{N}$ is the $N$-dimensional all-ones vector. $\bA \odot \bB$ denotes the Hadamard product of $\bA$ and $\bB$. \textbf{dom}$\left(f\right) $ is the domain of the function $f$. $\cN(\bm,\bR)$ is a complex Gaussian random vector with mean $\bm$ and covariance $\bR$. $\bbE\left[\cdot\right]$ is used to denote expectation.
%$\bA \circ \bB$ is the Khatri-Rao product of $\bA$, and $\bB$, $ \bA \otimes \bB$ is the Kronecker product of $\bA$, and $\bB$, and

%%%%%%%%%%%%%%%%%%%%%%%%%%%%%%%%%%%%%%%%%%%%%%%%%%%%%%%%%%%%%%%%%%%%%%%%%%%%%%%%%%%%%%%%%%%%%%%%%%%%%%%%%%%%%
\section{System and Channel Models} \label{sec:Model}
%%%%%%%%%%%%%%%%%%%%%%%%%%%%%%%%%%%%%%%%%%%%%%%%%%%%%%%%%%%%%%%%%%%%%%%%%%%%%%%%%%%%%%%%%%%%%%%%%%%%%%%%%%%%%
In this section, we describe the adopted frequency selective hybrid precoding system model and the wideband mmWave channel model. Key assumptions made for each model are also highlighted.

\subsection{System Model}
Consider the OFDM based system model in \figref{fig:Model} where a basestation (BS) with $N_\mathrm{BS}$ antennas and $N_\mathrm{RF}$ RF chains is assumed to communicate with a single mobile station (MS) with $N_\mathrm{MS}$ antennas and $N_\mathrm{RF}$ RF chains. The BS and MS communicate via $N_\mathrm{S}$ length-$K$ data symbol blocks, such that $N_\mathrm{S} \leq N_\mathrm{RF} \leq N_\mathrm{BS}$ and $N_\mathrm{S}\leq N_\mathrm{RF} \leq N_\mathrm{MS}$. In practice, the number of RF chains at the MS's is usually less than that of the BS's, but we do not exploit this fact in our model for simplicity of exposition. 
\begin{figure}[t]
	\centerline{
		\includegraphics[scale=.45]{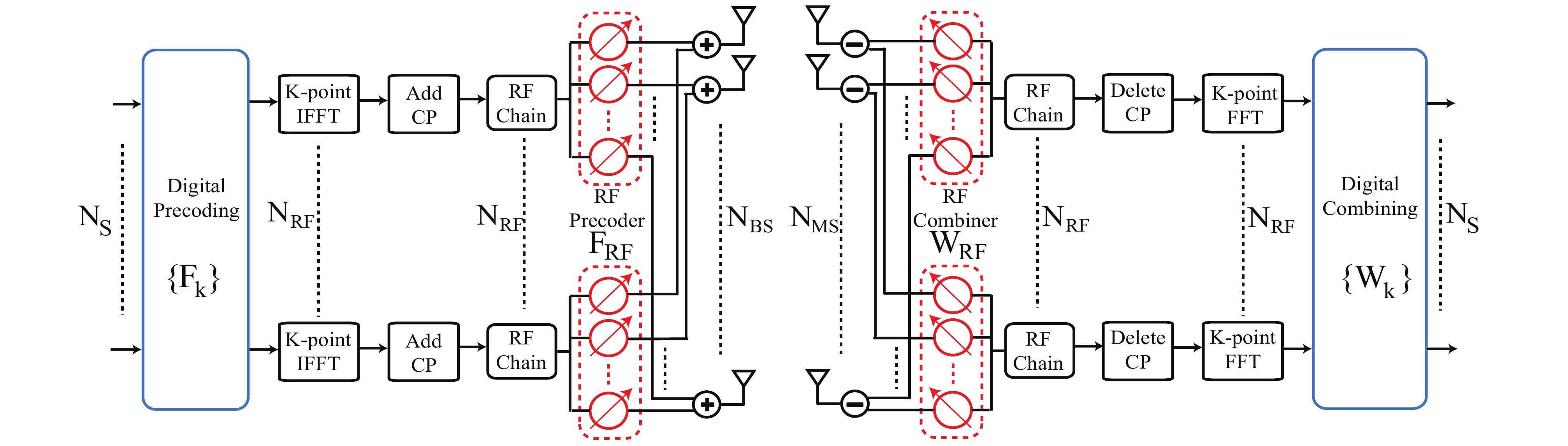}
	}
	\caption{A block diagram of the OFDM based BS-MS transceiver that employs hybrid analog/digital precoding.}
	\label{fig:Model}
\end{figure}

At the transmitter, the $N_\mathrm{S}$ data symbols $\bs_k$ at each subcarrier $k=1, ..., K$ are first precoded using an $N_\mathrm{RF} \times N_\mathrm{S}$ digital precoding matrix $\bF[k]$, and the symbol blocks are transformed to the time-domain using $N_\mathrm{RF}$ $K$-point IFFT's. Note that our model assumes that all subcarriers are used and, therefore, the data block length is equal to the number of subcarriers. A cyclic prefix of length $D$ is then added to the symbol blocks before applying the $N_\mathrm{BS} \times N_\mathrm{RF}$ RF precoding $\bF_\mathrm{RF}$. It is important to emphasize here that the RF precoding matrix $\bF_\mathrm{RF}$ is the same for \textit{all} subcarriers. This means that the RF precoder is assumed to be frequency flat while the baseband precoders can be different for each subcarrier. This is an important feature of the frequency selective hybrid precoding architecture in \figref{fig:Model} that differentiates it from the conventional OFDM-based unconstrained digital scheme where only frequency-selective digital precoders exist. The discrete-time transmitted complex baseband signal at subcarrier $k$ can therefore be written as
\begin{equation}
\bx[k]=\bF_\mathrm{RF} \bF[k] \bs[k],
\end{equation}
where $\bs[k]$ is the $N_\mathrm{S}\times 1$ transmitted vector at subcarrier $k$, such that  $\bbE\left[\bs[k]\bs^*[k]\right] = \frac{P}{K N_\mathrm{S}} \bI_{N_\mathrm{S}}$, and $P$ is the average total transmit power. Since $\bF_\mathrm{RF}$ is implemented using analog phase shifters, its entries are of constant modulus. To reflect that, we normalize the entries  $\left|\left[\bF_\mathrm{RF}\right]_{m,n}\right|^2=1$. Further, we assume that the angles of the analog phase shifters are quantized and have a finite set of possible values. With these assumptions, $\left[\bF_\mathrm{RF}\right]_{m,n}= e^{\j \phi_{m,n}}$, where $\phi_{m.n}$ is a quantized angle. The angle quantization assumption is discussed in more detail in \sref{sec:Codebook}.  Note that the RF beamforming can also be designed as a frequency selective filter \cite{Liang2007}, with additional hardware complexity. Two precoding power constraints are considered in this paper: (i) a total power constraint, where the hybrid precoders satisfy  $\sum_{k=1}^K \|\bF_\mathrm{RF} \bF[k]\|_F^2=K N_\mathrm{S}$, and (ii) a unitary power constraint, where the hybrid precoders meet $\bF_\mathrm{RF} \bF[k] \in \cU_{N_\mathrm{BS} \times N_\mathrm{S}}, k=1, 2, ..., K$, with the set of semi-unitary matrices $\cU_{N_\mathrm{BS} \times N_\mathrm{S}}=\left\{\bU \in \mathbb{C}^{N_\mathrm{BS} \times N_\mathrm{S}} | \bU^* \bU= \bI \right\}$. Note that while the total power constraint allows the transmit power to be distributed, possibly non-uniformly, among the subcarriers and the data streams on each subcarrier, the unitary power constraint enforces an equal power allocation among the subcarriers and the data streams on each subcarrier.  

At the MS, assuming perfect carrier and frequency offset synchronization, the received signal is first combined in the RF domain using the $N_\mathrm{MS} \times N_\mathrm{RF}$ combining matrix $\bW_\mathrm{RF}$. Then, the cyclic prefix is removed, and the symbols are returned back to the frequency domain with $N_\mathrm{RF}$ length-$K$ FFT's. The symbols at each subcarrier $k$ are then combined using the $N_\mathrm{RF} \times N_\mathrm{S}$ digital combining matrix $\bW[k]$. The constraints on the entries of RF combiner $\bW_\mathrm{RF}$ are similar to the RF precoders. Denoting the  $N_\mathrm{MS} \times N_\mathrm{BS}$ channel matrix at subcarrier $k$ as $\bH[k]$, the received signal at subcarrier $k$ after processing can be then expressed as
\begin{equation}
\by[k]=\bW^*[k] \bW_\mathrm{RF}^* \bH[k] \bF_\mathrm{RF} \bF[k] s[k]+ \bW^*[k] \bW_\mathrm{RF}^* \bn[k],
\label{eq:processed}
\end{equation}
where  $\bn[k] \sim \cN(\boldsymbol{0}, \sigma_\mathrm{N}^2 \bI)$ is the Gaussian noise vector corrupting the received signal.

\subsection{Channel Model}
To incorporate the wideband and limited scattering characteristics of mmWave channels \cite{Pi2011,Bai2014,11ad,Samimi2014,Rappaport2012,Rappaport2013a}, we adopt a geometric wideband mmWave channel model with $L$ clusters. Each cluster $\ell$ has a time delay $\tau_\ell \in \mathbb{R}$, and angles of arrival and departure (AoA/AoD), $\theta_\ell, \phi_\ell  \in [0, 2\pi]$. Each cluster $\ell$ is further assumed to contribute with $R_\ell$ rays/paths between the BS and MS \cite{Samimi2014,Forenza2007,Akdeniz2014}. Each ray $r_\ell=1, 2, ..., R_\ell,$ has a relative time delay $\tau_{r_\ell} $, relative AoA/AoD shift $\vartheta_{r_\ell}$, $\varphi_{r_\ell}$, and complex path gain $\alpha_{r_\ell}$. Further, let $\rho_\mathrm{PL}$ represent the path-loss between the BS and MS, and $p_\mathrm{rc}(\tau)$ denote a pulse-shaping function for $T_\mathrm{S}$-spaced signaling evaluated at $\tau$ seconds \cite{Schniter2014}. Under this model, the delay-$d$ MIMO channel matrix, $\boldsymbol{\mathit{H}}[d]$, can be written as \cite{Schniter2014}
\begin{equation}
\boldsymbol{\mathit{H}}[d]=\sqrt{\frac{N_\mathrm{BS}N_\mathrm{MS}}{\rho_\mathrm{PL}}}\sum_{\ell=1}^{L} \sum_{r_\ell=1}^{R_\ell} \alpha_{r_\ell} p_{\mathrm{rc}}\left(d T_\mathrm{S}-\tau_\ell -\tau_{r_\ell}\right) \ba_\mathrm{MS}\left(\theta_\ell - \mathcal{\vartheta}_{r_\ell}\right) \ba^*_\mathrm{BS} \left(\phi_\ell - \varphi_{r_\ell} \right),
\label{eq:d_Channel}
\end{equation}
where $\ba_\mathrm{BS}\left(\phi\right)$ and $\ba_\mathrm{MS}\left(\theta\right)$ are the antenna array response vectors of the BS and MS, respectively. Given the delay-$d$ channel model in \eqref{eq:d_Channel}, the channel at subcarrier $k$, $\bH[k]$, can be then expressed as \cite{Bajwa2010}
\begin{equation}
\bH[k]=\sum_{d=0}^{D-1} \boldsymbol{\mathit{H}}[d] e^{-j \frac{2 \pi k}{K} d}.
\label{eq:k_Channel}
\end{equation}

While most of the results developed in this paper are general for large MIMO channels, and not restricted to the channel model in \eqref{eq:d_Channel}, we described the wideband mmWave channel model in this section as it will be important for understanding the motivation behind the proposed construction of the hybrid analog/digital precoding codebooks in \sref{sec:Codebook}. Further, it will be adopted for the simulations in \sref{sec:Results} and for drawing conclusions about the performance of the proposed precoding schemes and codebooks in wideband mmWave channels. 

%%%%%%%%%%%%%%%%%%%%%%%%%%%%%%%%%%%%%%%%%%%%%%%%%%%%%%%%%%%%%%%%%%%%%%%%%%%%%%%%%%%%%%%%%%%%%%%%%%%%%%%%%%%%%
\section{Problem Statement} \label{sec:Form}
%%%%%%%%%%%%%%%%%%%%%%%%%%%%%%%%%%%%%%%%%%%%%%%%%%%%%%%%%%%%%%%%%%%%%%%%%%%%%%%%%%%%%%%%%%%%%%%%%%%%%%%%%%%%%

In this paper, we consider the downlink system model in \sref{sec:Model} when the BS and MS are connected via a limited feedback link. For this setup, we assume the MS has perfect channel knowledge with which it selects the best RF and baseband precoding matrices $\bF^\star_\mathrm{RF}$ and $\left\{\bF^\star \hspace{-2pt}\left[k\right]\right\}_{k=1}^K$ from predefined quantization codebooks to maximize the achievable mutual information when used by the BS. The main objective of this paper then is to develop efficient RF and baseband precoding codebooks for limited feedback wideband hybrid analog/digital precoding architectures. In this section, we first formulate the optimal hybrid precoding based mutual information when given RF and baseband precoding codebooks are used. Then, we briefly explain how the main objective of this paper will be investigated in the subsequent sections. 

As this paper focuses on the limited feedback hybrid precoding design, i.e., the design of $\bF_\mathrm{RF}, \left\{\bF[k]\right\}_{k=1}^K$, we will assume that the receiver can perform optimal nearest neighbor decoding based on the $N_\mathrm{MS}$-dimensional received signal with fully digital hardware. This allows decoupling the transceiver design problem, and focusing on the hybrid precoders design to maximize the mutual information of the system \cite{ElAyach2014}, defined as
\begin{equation}
\cI\left(\bF_\mathrm{RF}, \left\{\bF[k]\right\}_{k=1}^K\right)=\frac{1}{K} \sum_{k=1}^{K} \log_2 \left|\bI_{N_\mathrm{MS}}+\frac{\rho}{N_\mathrm{S}} \bH[k] \bF_\mathrm{RF} \bF[k]  \bF^*[k] \bF_\mathrm{RF}^*  \bH^*[k]  \right|,
\label{eq:MI}
\end{equation}
where $\rho=\frac{P}{K \sigma^2}$ is the SNR.  As combining with fully digital hardware, though, is not a practical mmWave solution, the hybrid combining design problem needs also to be considered. The design ideas that will be given in this paper for the hybrid precoders, however, provide direct tools for the construction of the hybrid combining matrices, $\bW_\mathrm{RF}$, $\left\{\bW[k]\right\}_{k=1}^K$, and is therefore omitted due to space limitations.

If the RF and baseband precoders are taken from quantized codebooks $\cF_\mathrm{RF}$ and $\cF_\mathrm{BB}$, respectively, then the maximum mutual information under the given hybrid precoding codebooks and the total power constraint is
\begin{equation}
\begin{aligned}
\cI^{\star}_\mathrm{HP} =  & \underset{\bF_\mathrm{RF}, \left\{\bF[k]\right\}_{ k=1}^K}  \max
& &\cI\left(\bF_\mathrm{RF}, \left\{\bF[k]\right\}_{k=1}^K\right) \\
& \hspace{20pt} \text{s.t.}
& &  \bF_\mathrm{RF } \in \cF_\mathrm{RF}, \\
&&& \bF[k] \in \cF_\mathrm{BB}, \ \ k=1, 2, ..., K, \\
&&& \sum_{k=1}^K \left\| \bF_\mathrm{RF} \bF[k] \right\|_F^2=K N_\mathrm{S}.
\label{eq:Opt_Feedback}
\end{aligned}
\end{equation}
The maximum mutual information with  hybrid precoding and under the unitary power constraint is similar but with the last constraint in \eqref{eq:Opt_Feedback} replaced with $\bF_\mathrm{RF} \bF[k] \in \cU_{N_\mathrm{BS} \times N_\mathrm{S}}$.

Our main objective in this work is to construct efficient hybrid precoding codebooks $\cF_\mathrm{RF}$ and $\cF_\mathrm{BB}$ to maximize the achievable mutual information in \eqref{eq:Opt_Feedback}. To get initial insights into the solution of this problem, we will first investigate a special case of the limited feedback hybrid precoding problem in \sref{sec:Optimal} when only the RF precoders are taken from quantized codebooks while no quantization constraints are imposed on the baseband precoders. For this problem, we will derive the optimal hybrid precoding design for any given RF codebook $\cF_\mathrm{RF}$. The results of \sref{sec:Optimal} will help us developing RF and baseband precoders codebook in \sref{sec:Codebook}. 

%%%%%%%%%%%%%%%%%%%%%%%%%%%%%%%%%%%%%%%%%%%%%%%%%%%%%%%%%%%%%%%%%%%%%%%%%%%%%%%%%%%%%%%%%%%%%%%%%%%%%%%%%%%%%
\section{Optimal Hybrid Precoding Design for a Given RF Codebook} \label{sec:Optimal}
%%%%%%%%%%%%%%%%%%%%%%%%%%%%%%%%%%%%%%%%%%%%%%%%%%%%%%%%%%%%%%%%%%%%%%%%%%%%%%%%%%%%%%%%%%%%%%%%%%%%%%%%%%%%%
% Points
In this section, we investigate the limited feedback hybrid precoding design when only the RF precoders are taken from quantized codebooks. This problem is of a special interest for two main reasons. First, it will provide useful insights into the construction of efficient hybrid analog/digital precoding codebooks as will be summarized at the end of this section. Second, the hybrid precoding design problem with only RF precoders quantization can also be interpreted as the hybrid precoding design problem with perfect channel knowledge. The reason is that even when perfect channel knowledge is available at the transmitter, the RF precoders will be taken from a certain codebook that captures the hardware constraints such as the phase shifters quantization. With this motivation, we consider the following relaxation of the optimization in \eqref{eq:Opt_CSI} that captures the assumption that only the RF precoders are quantized. 
\begin{equation}
\begin{aligned}
\cI^{\star}_\mathrm{HP} =  & \underset{\bF_\mathrm{RF}, \left\{\bF[k]\right\}_{ k=1}^K}  \max
& &\cI\left(\bF_\mathrm{RF}, \left\{\bF[k]\right\}_{k=1}^K\right) \\
& \hspace{20pt} \text{s.t.}
& &  \bF_\mathrm{RF } \in \cF_\mathrm{RF}, \\
&&& \sum_{k=1}^K \left\| \bF_\mathrm{RF} \bF[k] \right\|_F^2=K N_\mathrm{S}.
\label{eq:Opt_CSI}
\end{aligned}
\end{equation}

The design of the hybrid analog/digital precoders in \eqref{eq:Opt_CSI} is non-trivial due to (i) the RF hardware non-convex constraint $\bF_\mathrm{RF}\in \cF_\mathrm{RF}$, and (ii) the coupling between the analog and digital precoding matrices, which arises in the power constraint (the second constraint in \eqref{eq:Opt_CSI}). Due to these difficulties, prior work \cite{ElAyach2014,Kim2013,Alkhateeb2014} focused on developing heuristics designs for the hybrid analog/digital precoders in \eqref{eq:Opt_CSI}. Although these heuristic algorithms were shown to give good performance, they do not provide enough insights that help, for example, to design limited feedback hybrid precoding codebooks.

In this section, we will consider the coupling between the analog and digital precoders, and show that the optimal baseband precoders can be written as a function of the RF precoders under both the total and unitary power constraints. This will reduce the hybrid precoding design problem to an RF precoder design problem. 
%---------------------------------------------------------------
\subsection{Total Power Constraint} \label{subsec:TP}
%--------------------------------------------------------------
As the RF precoding matrix $\bF_\mathrm{RF}$ in \eqref{eq:Opt_CSI} is taken from a quantized codebook $\cF_\mathrm{RF}$, then the optimal mutual information in \eqref{eq:Opt_CSI} can also be equivalently written in the following outer-inner problems form
\begin{equation}
\cI^{\star}_\mathrm{HP} = \underset{\bF_\mathrm{RF} \in \cF_\mathrm{RF}}  \max \  
\left\{ \begin{alignedat}{2}
& \underset{ \left\{\bF[k]  \right\}_{ k=1}^K} \max \hspace{10pt} \cI\left(\bF_\mathrm{RF}, \left\{\bF[k]\right\}_{k=1}^K\right)\\
& \hspace{10pt} \text{s.t.} \hspace{15pt} \sum_{k=1}^K \left\| \bF_\mathrm{RF} \bF[k] \right\|_F^2=K N_\mathrm{S},
\end{alignedat}\right\}
\label{eq:Opt_CSI_IO}
\end{equation}
where the outer maximization is over the set of possible quantized RF precoding matrices, and the inner problem is over the set of feasible baseband precoders given the RF precoder, $\left\{\bF[k] \in \mathbb{C}^{N_\mathrm{RF} \times N_\mathrm{S}}|  \sum_{k=1}^K \left\| \bF_\mathrm{RF} \bF[k] \right\|_F^2=K N_\mathrm{S} \right\}$.

Note that the solution of the optimal baseband precoders in the inner problem of  \eqref{eq:Opt_CSI_IO} is not given by the simple SVD of the effective channel with the RF precoder, $\bH[k] \bF_\mathrm{RF}$, because of the different power constraint that represents the coupling between the baseband and RF precoders. In the following proposition, we find the optimal baseband precoders of the inner problem of \eqref{eq:Opt_CSI_IO}.
\begin{proposition}
Define the SVD decompositions of the $k$th subcarrier channel matrix $\bH[k]$ as $\bH[k]=\bU[k] \boldsymbol{\Sigma}[k] \bV^*[k]$, and the SVD decomposition of the matrix  $\boldsymbol{\Sigma}[k] \bV^*[k] \bF_\mathrm{RF} \left(\bF_\mathrm{RF}^* \bF_\mathrm{RF} \right)^{-\frac{1}{2}}$ as $\boldsymbol{\Sigma}[k] \bV^*[k] \bF_\mathrm{RF} \left(\bF_\mathrm{RF}^* \bF_\mathrm{RF} \right)^{-\frac{1}{2}}= \overline{\bU}[k] \overline{\boldsymbol{\Sigma}}[k] \overline{\bV}^*[k]$. Then, the baseband precoders $\left\{\bF[k]\right\}_{k=1}^K$ that solve the inner optimization problem of \eqref{eq:Opt_CSI_IO} are given by
\begin{equation}
\bF^\star[k]= \left(\bF_\mathrm{RF}^* \bF_\mathrm{RF} \right)^{-\frac{1}{2}} \left[\overline{\bV}[k]\right]_{:,1:N_\mathrm{S}} \boldsymbol{\Lambda}[k],\ \ k=1, 2, ..., K,
\label{eq:Opt_BB_TP}
\end{equation}
where $\left[\overline{\bV}[k]\right]_{:,1:N_\mathrm{S}}$ is the $N_\mathrm{RF} \times N_\mathrm{S}$ matrix that gathers the $N_\mathrm{S}$ dominant vectors of $\overline{\bV}[k]$, and $\boldsymbol{\Lambda}[k]$ is an $N_\mathrm{S} \times N_\mathrm{S}$ water-filling power allocation diagonal matrix with
\begin{equation}
\left[\boldsymbol{\Lambda}[k]\right]_{i,i}^2=\left(\mu-\frac{N_\mathrm{S}}{\rho \left[\overline{\boldsymbol{\Sigma}}[k]\right]^2}\right)^{+}, i=1, ..., N_\mathrm{S}, \ k=1, ..., K, 
\end{equation}
and with $\mu$ satisfying
\begin{equation}
\sum_{k=1}^K \sum_{i=1}^{N_\mathrm{S}} \left(\mu-\frac{N_\mathrm{S}}{\rho \left[\overline{\boldsymbol{\Sigma}}[k]\right]^2}\right)^{+} = K Ns
\end{equation}

%\begin{equation}
%\left[\boldsymbol{\Lambda}_\mathrm{T}\right]_{i,i}^2=\left(\frac{N_\mathrm{S}\left(K \rho+\sum_{n=1}^{\bar{N}_\mathrm{S}} \left[\overline{\boldsymbol{\Sigma}}_\mathrm{T} \right]^{-2}_{n,n}\right)}{\rho \bar{N}_\mathrm{S} }-\frac{N_\mathrm{S}}{\rho \left[\overline{\boldsymbol{\Sigma}}_\mathrm{T}\right]^2_{i,i}}\right)^{+}, \ \ i=1, 2, ..., K N_\mathrm{S},
%\end{equation}
%
%where $\boldsymbol{\Lambda}_\mathrm{T}=\mathrm{blkdiag}\left(\boldsymbol{\Lambda}[1], ..., \boldsymbol{\Lambda}[K]\right)$, $\boldsymbol{\Lambda}_\mathrm{T}=\mathrm{blkdiag}\left(\boldsymbol{\Lambda}[1], ..., \boldsymbol{\Lambda}[K]\right)$, and $\bar{N}_\mathrm{S}< N_\mathrm{S}$, is the number of positive $\left[\boldsymbol{\Lambda}_\mathrm{T}\right]_{i,i}^2$.
\label{prop:Opt_TP}
\end{proposition}
\begin{proof}
See Appendix \ref{app:prop_OPT_TP}.
\end{proof}

Given the optimal baseband precoder in \eqref{eq:Opt_BB_TP}, the optimal hybrid precoding based mutual information with the RF codebook $\cF_\mathrm{RF}$ and a total power constraint can now be written as
\begin{equation}
\cI^{\star}_\mathrm{HP} =  \underset{\bF_\mathrm{RF} \in \cF_\mathrm{RF}}  \max \frac{1}{K} \sum_{k=1}^{K} \log_2 \left|\bI_{N_\mathrm{S}}+\frac{\rho}{N_\mathrm{S}} \left[\overline{\boldsymbol{\Sigma}}[k]\right]_{1:N_\mathrm{S}, 1:N_\mathrm{S}}^2 \boldsymbol{\Lambda}[k]^2  \ \right|,
\label{eq:Opt_MI_TP}
\end{equation}
where $\overline{\boldsymbol{\Sigma}}[k]^2$ and  $\boldsymbol{\Lambda}[k]^2$ are functions only of $\bF_\mathrm{RF}$ and $\bH[k]$ as defined in Proposition \ref{prop:Opt_TP}. This means that the optimal hybrid precoding based mutual information is determined only by the RF precoders design. Hence, an exhaustive search over the RF precoders codebook $\cF_\mathrm{RF}$ is sufficient to find the maximum achievable mutual information with hybrid precoding. 
%This observation is useful for many reasons as will be summarized at the end of this section.

The hybrid precoding design in Proposition \ref{prop:Opt_TP} can also be extended to the case when the power constraint is imposed on each subcarrier. In this case, the power constraint on the hybrid precoders is written as $\left\|\bF_\mathrm{RF} \bF[k]\right\|_\mathrm{F}=N_\mathrm{S}, k=1, ..., K$. The following corollary presents the optimal baseband precoder for a given RF codebook, under the per-subcarrier total power constraint. 

\begin{corollary}\label{cor:OPT_TPS}
	The optimal baseband precoders that maximizes the objective of the inner optimization problem of \eqref{eq:Opt_CSI_IO}, under the constraint $\left\|\bF_\mathrm{RF} \bF[k]\right\|_\mathrm{F}=N_\mathrm{S}, k=1, ..., K$, are given by
	\begin{equation}
	\bF^\star[k]= \left(\bF_\mathrm{RF}^* \bF_\mathrm{RF} \right)^{-\frac{1}{2}} \left[\overline{\bV}[k]\right]_{:,1:N_\mathrm{S}} \boldsymbol{\Lambda}_\mathrm{P}[k],\ \ k=1, 2, ..., K,
	\label{eq:Opt_BB_TPP}
	\end{equation}
	where $\boldsymbol{\Lambda}_\mathrm{P}[k]$ is an $N_\mathrm{S} \times N_\mathrm{S}$ water-filling power allocation diagonal matrix with
	\begin{equation}
	\left[\boldsymbol{\Lambda}_\mathrm{P}[k]\right]_{i,i}^2=\left(\mu-\frac{N_\mathrm{S}}{\rho \left[\overline{\boldsymbol{\Sigma}}[k]\right]^2}\right)^{+}, i=1, ..., N_\mathrm{S}, 
	\end{equation}
	and with $\mu$ satisfying
	\begin{equation}
	\sum_{i=1}^{N_\mathrm{S}} \left(\mu-\frac{N_\mathrm{S}}{\rho \left[\overline{\boldsymbol{\Sigma}}[k]\right]^2}\right)^{+} = Ns, \ k=1, ..., K.
	\end{equation} 
\end{corollary}
\begin{proof}
	The proof is similar to Proposition \ref{prop:Opt_TP} and is therefore omitted. 
\end{proof}

It is worth mentioning here that most current wireless systems do not perform per subcarrier power allocation. This constraint, therefore, is not especially critical for practical systems.  

An important note on the structure of the optimal hybrid precoders derived in Proposition \ref{prop:Opt_TP} and corollary \ref{cor:OPT_TPS} is that the matrix $\bF_\mathrm{HP}[k]$ representing this optimal hybrid precoders at subcarrier $k$ can be written as 
\begin{equation}\label{eq:Tot_Unit_Str}
\bF_\mathrm{HP}[k]=\bF_\mathrm{U}[k] \Lambda[k],
\end{equation} 
where $\bF_\mathrm{U}[k] = \bF_\mathrm{RF}^\star \left(\left(\bF_\mathrm{RF}^\star\right)^* \bF_\mathrm{RF}^\star\right)^{-\frac{1}{2}} \left[\overline{\bV}[k]\right]_{:,1:N_\mathrm{S}}$ is a semi-unitary matrix, as it can be verified that $\bF_\mathrm{U}^*[k] \bF_\mathrm{U}[k]=\bI_{N_\mathrm{S}}$. This means that the structure of the optimal hybrid precoders is similar to that of the unconstrained SVD precoders, as it is written as a product of a semi-unitary matrix and a diagonal water-filling power allocation matrix. 

%-------------------------------------------------------------
\subsection{Unitary Power Constraint}
%--------------------------------------------------------------

For limited feedback MIMO systems, the unitary power constraint which requires the columns of the precoding matrix $\bF_\mathrm{RF} \bF[k]$ to be orthogonal with equal power, is an alternative important constraint. Even though some performance loss should be expected with unitary constraints compared with the more relaxed total power constraint, unitary constraints usually lead to more efficient codebooks and codeword selection algorithms for limited feedback systems \cite{Love2008}. Further, they normally offer a close performance to the total power constraint \cite{Love2008}. In this subsection, we investigate the optimal hybrid precoding design under a unitary power constraint, and conclude important results for limited feedback hybrid precoding.

Similar to \eqref{eq:Opt_CSI_IO}, the optimal mutual information with hybrid precoding under the unitary power constraint can be written in the following outer-inner problems form
\begin{equation}
\cI^{\star}_\mathrm{HP} =   \underset{\bF_\mathrm{RF} \in \cF_\mathrm{RF}}  \max \ \ 
\left\{ \begin{alignedat}{2}
& \underset{ \left\{\bF[k]  \right\}_{ k=1}^K} \max \hspace{10pt} \cI\left(\bF_\mathrm{RF}, \left\{\bF[k]\right\}_{k=1}^K\right) \\
& \hspace{10pt} \text{s.t.}  \hspace{15pt}  \bF_\mathrm{RF} \bF[k]  \in \cU_{N_\mathrm{BS} \times N_\mathrm{S}}, \ \ k=1, 2, ..., K.
\end{alignedat}\right\}
\label{eq:Opt_CSI_IO_U}
\end{equation}
Given an RF precoder $\bF_\mathrm{RF}$, we find, in the following proposition, the optimal baseband precoders of the inner problem of \eqref{eq:Opt_CSI_IO_U}.
\begin{proposition}
Define the SVD decompositions of the $k$th subcarrier channel matrix $\bH[k]$ and the matrix $\boldsymbol{\Sigma[k]} \bV^*[k] \left(\bF_\mathrm{RF}^* \bF_\mathrm{RF} \right)^{-\frac{1}{2}}$ as in Proposition \ref{prop:Opt_TP}, then the baseband precoders $\left\{\bF[k]\right\}_{k=1}^K$ that solve the inner optimization problem of \eqref{eq:Opt_CSI_IO_U} are given by
\begin{equation}
\bF^\star[k]= \left(\bF_\mathrm{RF}^* \bF_\mathrm{RF} \right)^{-\frac{1}{2}} \left[\overline{\bV}[k]\right]_{:,1:N_\mathrm{S}}, \ \ k=1, 2, ..., K. \label{eq:Opt_BB_UP}
\end{equation}
\label{prop:Opt_UP}
\end{proposition}
\begin{proof}
The proof is similar to that in Appendix \ref{app:prop_OPT_TP}, and is skipped due to space limitations.
\end{proof}

Given the optimal baseband precoder in \eqref{eq:Opt_BB_UP}, the optimal hybrid precoding based mutual information with the RF codebook $\cF_\mathrm{RF}$ and the unitary power constraint can be written as
\begin{equation}
\cI^{\star}_\mathrm{HP} =  \underset{\bF_\mathrm{RF} \in \cF_\mathrm{RF}}  \max \frac{1}{K} \sum_{k=1}^{K} \log_2 \left|\bI_{N_\mathrm{S}}+\frac{\rho}{N_\mathrm{S}} \left[\overline{\boldsymbol{\Sigma}}[k]\right]_{1:N_\mathrm{S}, 1:N_\mathrm{S}}^2  \ \right|,
\label{eq:Opt_MI_UP}
\end{equation}
where $\overline{\boldsymbol{\Sigma}}[k]^2$ depends only on $\bF_\mathrm{RF}$ and $\bH[k]$ as defined in Proposition \ref{prop:Opt_TP}. Next, we state an important remark on the structure of the optimal hybrid precoding design.\\
\hspace{-15pt} \textbf{Remark 1.} \textit{The optimal baseband precoder $\bF^\star[k]$ under the unitary hybrid precoding power constraint $\bF_\mathrm{RF} \bF[k] \in \cU_{N_\mathrm{BS} \times N_\mathrm{S}}$ is decomposed as $\bF^\star[k]=\bA_\mathrm{RF} \bG^\star[k]$, where $\bA_\mathrm{RF}=\left(\bF_\mathrm{RF}^* \bF_\mathrm{RF}\right)^{-\frac{1}{2}}$ depends only on the RF precoder, and $\bG^\star[k]$, which we call the \textit{equivalent} baseband precoder, is a semi-unitary matrix, $\bG^\star[k] \in \cU_{N_\mathrm{RF} \times N_\mathrm{S}}$, with the optimal design described in \eqref{eq:Opt_BB_UP}}.\\
Remark 1 shows that for the BS to achieve the optimal mutual information with the unitary power constraint and RF codebook $\cF_\mathrm{RF}$, it needs to know (i) the index of the RF precoder codeword that solves \eqref{eq:Opt_MI_UP}, and (ii) the optimal \textit{semi-unitary} equivalent baseband precoding matrix $\bG^\star[k]$.

Although an exhaustive search over the RF codebook is still required to find the optimal mutual information in \eqref{eq:Opt_MI_TP} and \eqref{eq:Opt_MI_UP}, the derived results are useful for several reasons. First, equations \eqref{eq:Opt_MI_TP} and \eqref{eq:Opt_MI_UP} provide, for the first time, the maximum achievable rate with hybrid precoding for any given RF codebooks. Therefore, these equations give a benchmark that can be used to evaluate the performance of any heuristic/iterative hybrid analog/digital precoding algorithms, and to estimate how much additional improvement is possible. Further, the optimal mutual information in \eqref{eq:Opt_MI_TP} and \eqref{eq:Opt_MI_UP}, depend only on the RF codebook, which will help in the design of the RF codebook as we will see in \sref{sec:Codebook}.  Another useful finding is the special construction of the optimal baseband precoder described in Remark 1 which offers insights into the limited feedback hybrid precoding design as will be described in \sref{sec:Codebook}.

For the remaining part of this paper, we will focus on the hybrid precoding design problem with the unitary constraint. In the next section, we will address the design of hybrid precoding codebooks for limited feedback wideband mmWave systems. Then, in \sref{sec:Greedy}, we will develop a greedy frequency selective hybrid precoding algorithm based on Gram-Schmidt orthogonalization, that relaxes the exhaustive search requirement over the RF codebook in \eqref{eq:Opt_MI_UP} while providing a near-optimal performance. It is worth noting here that as shown in \eqref{eq:Tot_Unit_Str}, the optimal hybrid precoders under total power constraints consist of a semi-unitary matrix multiplied by a diagonal water-filling power allocation matrix. Therefore, the hybrid precoding codebooks and codeword selection algorithms that will be designed under unitary power constraints can also be used for hybrid precoding under total power constraints. The water=filling power allocation can be done as a subsequent step to further improve the performance.

%%%%%%%%%%%%%%%%%%%%%%%%%%%%%%%%%%%%%%%%%%%%%%%%%%%%%%%%%%%%%%%%%%%%%%%%%%%%%%%%%%%%%%%%%%%%%%%%%%%%%%%%%%%%%
\section{Codebook Design for Frequency Selective Hybrid Precoding} \label{sec:Codebook}
%%%%%%%%%%%%%%%%%%%%%%%%%%%%%%%%%%%%%%%%%%%%%%%%%%%%%%%%%%%%%%%%%%%%%%%%%%%%%%%%%%%%%%%%%%%%%%%%%%%%%%%%%%%%%

In this section, we consider the wideband mmWave system model in \sref{sec:Model} with limited feedback, and develop hybrid analog and digital codebooks. First, we will consider the case  $N_\mathrm{S}=N_\mathrm{RF}$ in \sref{subsec:Nrf} where we leverage the structure of the optimal hybrid precoders developed in \sref{sec:Optimal} to show that the hybrid codebook design problem can be reduced to an RF codebook design problem. Then, we consider the case $N_\mathrm{S}<N_\mathrm{RF}$ in \sref{subsec:Ns}, where we develop hybrid analog and digital precoding codebooks leveraging the results  in \sref{subsec:Nrf}. 

%------------------------------------------------------------
\subsection{Case 1: $N_\mathrm{S}=N_\mathrm{RF}$} \label{subsec:Nrf}
%------------------------------------------------------------

Given the optimal baseband precoder structure from \eqref{eq:Opt_BB_UP}, the optimal hybrid precoding based mutual information when $N_\mathrm{S}=N_\mathrm{RF}$, and with the RF codebook $\cF_\mathrm{RF}$ can be written as
\begin{align}
& \cI^{\star}_\mathrm{HP} = \nonumber \\
& \underset{\bF_\mathrm{RF} \in \cF_\mathrm{RF}}  \max \frac{1}{K} \sum_{k=1}^{K} \log_2 \left|\bI_{{r}(\bH[k])}+\frac{\rho}{N_\mathrm{S}} \boldsymbol{\Sigma}[k]^2 \bV^*[k] \bF_\mathrm{RF} \left(\bF_\mathrm{RF}^* \bF_\mathrm{RF}\right)^{-\frac{1}{2}} \bG[k] \bG^*[k]  \left(\bF_\mathrm{RF}^* \bF_\mathrm{RF}\right)^{-\frac{1}{2}} \bF_\mathrm{RF}^* \bV[k]  \right|.
\label{eq:Opt_MI_UP_Nrf}
\end{align}
Since $\bG[k]$ is unitary for $N_S=N_\mathrm{RF}$, equation \eqref{eq:Opt_MI_UP_Nrf} can be equivalently written as
\begin{equation}
\cI^{\star}_\mathrm{HP} =  \underset{\bF_\mathrm{RF} \in \cF_\mathrm{RF}}  \max \frac{1}{K} \sum_{k=1}^{K} \log_2 \left|\bI_{r(\bH[k])}+\frac{\rho}{N_\mathrm{S}} \boldsymbol{\Sigma}[k]^2 \bV^*[k] \bF_\mathrm{RF} \left(\bF_\mathrm{RF}^* \bF_\mathrm{RF}\right)^{-1} \bF_\mathrm{RF} \bV[k]  \right|.
\label{eq:Opt_MI_UP_Nrf_I}
\end{equation}
As a result, the optimal mutual information is invariant to the $optimal$ equivalent baseband precoder, and depends only on the knowledge of the RF precoder $\bF_\mathrm{RF}$. This leads to the following remark.\\
\hspace{-15pt} \textbf{Remark 2.} \textit{With $N_\mathrm{S}=N_\mathrm{RF}$, feeding back only the index of the optimal RF precoder that solves \eqref{eq:Opt_MI_UP_Nrf_I} is sufficient to achieve the optimal mutual information with limited feedback hybrid precoding.}

Remark 2 also means that no quantization of the baseband precoder is required when $N_\mathrm{S}=N_\mathrm{RF}$. Further, optimizing the limited feedback hybrid precoding performance is achieved by the optimization of the RF codebook design $\cF_\mathrm{RF}$, which is addressed in the remaining part of this subsection.

\textbf{RF Codebook Design Criterion:}
Our objective is to design the RF codebook to minimize the distortion given by the average  mutual information loss of hybrid precoding compared with the optimal unconstrained per-subcarrier SVD solution. Denoting the SVD of the RF precoder as $\bF_\mathrm{RF}=\bU_\mathrm{RF} \boldsymbol{\Sigma}_\mathrm{RF} \bV_\mathrm{RF}^*$, the optimal mutual information with limited feedback hybrid precoding in \eqref{eq:Opt_MI_UP_Nrf_I} can be written as
\begin{equation}
\cI^{\star}_\mathrm{HP} =  \underset{\bF_\mathrm{RF} \in \cF_\mathrm{RF}}  \max \frac{1}{K} \sum_{k=1}^{K} \log_2 \left|\bI_{r(\bH[k])}+\frac{\rho}{N_\mathrm{S}} \boldsymbol{\Sigma}[k]^2 \bV^*[k] \bU_\mathrm{RF} \bU_\mathrm{RF}^* \bV[k]  \right|.
\end{equation}

For large mmWave MIMO systems, a reasonable assumption as stated in \cite{ElAyach2014} for narrowband channels is that the hybrid precoders can be made sufficiently close to the dominant channel eigenspace. Further, the dominant  channel eigenspaces of the different subcarriers may have high correlation at mmWave channels \cite{Rappaport2013a,Ghosh2014}. This means that the eigenvalues of the matrix $\bI- \tilde{\bV}^*[k] \bF_\mathrm{RF} \bF[k] \bF^*[k] \bF_\mathrm{RF}^* \tilde{\bV}[k]$ can be made sufficiently small.  Using this assumption, which will also be evaluated by simulations in \figref{fig:Fig2}-\figref{fig:Fig4x}, and following  similar steps to that in equations (12)-(14) of \cite{ElAyach2014}, $\cI^\star_\mathrm{HP}$ can be approximately written as
\begin{align}
\cI^{\star}_\mathrm{HP} &\approx  \underset{\bF_\mathrm{RF} \in \cF_\mathrm{RF}}  \max \frac{1}{K} \sum_{k=1}^{K} \left( \log_2 \left|\bI_{N_\mathrm{RF}}+\frac{\rho}{N_\mathrm{S}} \tilde{\boldsymbol{\Sigma}}[k]^2 \right| - \left(N_\mathrm{RF} - \left\|\bU_\mathrm{RF}^* \tilde{\bV}[k] \right\|_F^2\right) \right), \label{eq:Large_Approx}\\
&= \frac{1}{K} \sum_{k=1}^{K} \log_2 \left|\bI_{N_\mathrm{RF}}+\frac{\rho}{N_\mathrm{S}} \tilde{\boldsymbol{\Sigma}}[k]^2 \right| - \underset{\bF_\mathrm{RF} \in \cF_\mathrm{RF}}  \min \frac{1}{K}\sum_{k=1}^K \left(N_\mathrm{RF} - \left\|\bU_\mathrm{RF}^* \tilde{\bV}[k] \right\|_F^2 \right),
\end{align}
where $\tilde{\boldsymbol{\Sigma}}[k]=\left[\boldsymbol{\Sigma}[k]\right]_{1:N_\mathrm{RF}, 1:N_\mathrm{RF}}$ and $\tilde{\bV}[k]=\left[\bV[k]\right]_{:,1:N_\mathrm{RF}}$.

When fully digital unconstrained precoding  with perfect channel knowledge is possible, the optimal mutual information with per-subcarrier unitary constraint is achieved by per-subcarrier SVD precoding and is equal to
\begin{equation}
\cI^\star_{\mathrm{UC}}=\frac{1}{K} \sum_{k=1}^K \log_2 \left|\bI_{N_\mathrm{RF}}+\frac{\rho}{N_\mathrm{S}}\tilde{\boldsymbol{\Sigma}}[k]^2 \right|.
\label{eq:Opt_MI_UC}
\end{equation}

We can now define the distortion due to limited feedback hybrid precoding with the RF precoder $\cF_\mathrm{RF}$ as
\begin{align}
\cD\left(\cF_\mathrm{RF}\right)&=\bbE_{\left\{\bH[k]\right\}_{k=1}^K}\left[\cI_\mathrm{UC}^\star-\cI^\star_\mathrm{HP}\right],\\
& \approx \bbE_{\left\{\bH[k]\right\}_{k=1}^K}\left[ \underset{\bF_\mathrm{RF} \in \cF_\mathrm{RF}}  \min \frac{1}{K}\sum_{k=1}^K \left(N_\mathrm{RF} - \left\|\bU_\mathrm{RF}^* \tilde{\bV}[k] \right\|_F^2 \right)\right], \label{eq:Approx_Dist}\\
& = \bbE_{\left\{\bH[k]\right\}_{k=1}^K}\left[ \underset{\bF_\mathrm{RF} \in \cF_\mathrm{RF}}  \min \frac{1}{K}\sum_{k=1}^K d^2_\mathrm{chord}\left(\bU_\mathrm{RF},\tilde{\bV}[k] \right)\right],\\
& = \bbE_{\left\{\bH[k]\right\}_{k=1}^K}\left[ \underset{\bF_\mathrm{RF} \in \cF_\mathrm{RF}}  \min \Phi_\mathrm{chord} \left(\bU_\mathrm{RF},\left\{\tilde{\bV}[k]\right\}_{k=1}^K \right)\right],
\label{eq:Dist}
\end{align}
where $d_\mathrm{chord}\left(\bX,\bY\right)$ is the chordal distance between the two points $\bX, \bY$ on the Grassmann manifold $\cG\left(N_\mathrm{BS}, N_\mathrm{RF}\right)$. $\Phi_\mathrm{chord}\left(\bX,\left\{\bY[k]\right\}_{k=1}^K\right)$ is the average of the squared chordal distances between the Grassmann manifold points $\bX$ and $\left\{\bY[k]\right\}_{k=1}^K$. If no constraints are imposed on $\bX$, the solution of $\arg\min_{\bX \in \cG\left(N_\mathrm{BS},N_\mathrm{RF}\right)} \Phi_\mathrm{chord}\left(\bX,\left\{\bY[k]\right\}_{k=1}^K\right)$ is given by the Karcher mean of the $K$ $N_\mathrm{RF}$-dimensional subspaces defined by the points $\left\{\bY[k]\right\}_{k=1}^K$\cite{Absil2004}. Our RF codebook design criterion is then to minimize the  distortion function expression in \eqref{eq:Dist}

\textbf{RF Codebook Construction:}
Developing a closed-form solution for the RF precoders codebook that minimizes the distortion in \eqref{eq:Dist} is non-trivial for two main reasons. First, the RF hardware constraints like the constant modulus limitation on the entries of the RF precoding matrix and the angle quantization of the phase shifters, which impose non-convex constraints on the distortion function minimization problem. Second, the lack of knowledge about the closed-form distributions of the mmWave channel matrices. These closed-form distributions usually play a key role in constructing the precoders codebook. For example, the uniform distribution of the dominant singular vectors of the IID complex Gaussian MIMO channels led to the codebook design based on isotropic packing of the Grassmann manifold \cite{Love2008,Love2005}.

To overcome these challenges, we developed Algorithm \ref{alg:RF_CB} which is a Lloyd-type algorithm \cite{Lloyd1982,Pitaval2014}, that first constructs a precoders codebook to minimize the distortion function in \eqref{eq:Dist} for wideband mmWave channels while neglecting the RF hardware constraints. Then, the RF precoding codewords are designed to minimize the additional distortion results from the RF hardware constraints. One advantage of developing a Lloyd-type algorithm is that no knowledge about the closed-form distributions of the channel matrices is required, and only the knowledge of the mmWave channel parameter statistics, which are given by measurements \cite{Samimi2014}, is needed. These parameter statistics are used to generate random channel realizations, according to \eqref{eq:d_Channel}, which are employed in constructing the RF precoders codebook as described in Algorithm \ref{alg:RF_CB}. 

\begin{algorithm} [!t]                     % enter the algorithm environment
\caption{RF Codebook Construction for Frequency Selective Hybrid Precoding}          % give the algorithm a caption
\label{alg:RF_CB}                           % and a label for \ref{} commands later in the document
\begin{algorithmic} %[1]                   % enter the algorithmic environment
    \State 1) \textbf{Initialization:} Generate random initial centroid points $\cF_\mathrm{U}=\left\{\bF_1^\mathrm{U}, ..., \bF_{N_\mathrm{CB}}^\mathrm{U} \right\} \subset \cU_{N_\mathrm{BS} \times N_\mathrm{S}}$.
    \State 2) \textbf{Source:} Generate random mmWave channels according to \eqref{eq:k_Channel}, $\cH=\left\{\left\{\bH[k]\right\}_{k=1}^K\right\}$, and construct the set of dominant right singular vectors corresponding to the generated channels $\cV$ where
    \begin{equation}
    \cV=\left\{\left\{\tilde{\bV}[k]\right\}_{k=1}^K\subset \cU_{N_\mathrm{BS}\times N_\mathrm{S}} \left| \tilde{\bV}[k]=\left[\bV[k]\right]_{:,1:N_\mathrm{S}}, \bH[k]=\bU[k] \boldsymbol{\Sigma}[k] \bV[k], \left\{\bH[k]\right\}_{k=1}^K \in \cH\right.\right\}.\nonumber
    \end{equation}
    \State 3) \textbf{Nearest Neighbor Partitioning:} Partition the set $\cV$ into $N_\mathrm{CB}$ Voronoi cells $\left\{\cR_1, ..., \cR_{N_\mathrm{CB}}\right\}$ according to \eqref{eq:centroid}-\eqref{eq:Cell}.
    \State 4) For each Voronoi cell $n$, $n=1, ..., N_\mathrm{CB}$
    \State \hspace{20pt} a) \textbf{Karcher Mean Calculation:} Calculate the Karcher mean $\bM_n$ of the points \\
    \hspace{33pt} $\left\{\left\{\tilde{\bV}[k]\right\}_{k=1}^K\right\}$ in $\cR_n$ according to \eqref{eq:cent_eig}.
    \State \hspace{20pt} b) \textbf{Updating the Centroid:} Update the $n$th unconstrained codeword $\bF_n^\mathrm{U}=\bM_n$.
    \State \hspace{20pt} c) \textbf{RF Codeword Approximation:} Calculate the approximated RF codeword $\bF_n^\mathrm{RF}$ \\
    \hspace{33pt} according to  \eqref{eq:RF_approx_ele}.
    \State 5) Loop back to step 3) until convergence
\end{algorithmic}
\end{algorithm}

The operation of Algorithm \ref{alg:RF_CB} can be summarized as follows.
\begin{itemize}
\item{Initialization and source generation:}{ In this step, $N_\mathrm{CB}$ initial codewords $\bF_n^\mathrm{U}, n=1,..., N_\mathrm{CB}$ for the unconstrained codebook are randomly chosen from $\cU_{N_\mathrm{BS}\times N_\mathrm{S}}$. Further, random wideband mmWave channel realizations are generated according to \eqref{eq:d_Channel}, \eqref{eq:k_Channel} with the parameter statistics given from measurements, e.g., \cite{Samimi2014}. For each channel realization, the $K$ subcarrier channels $\left\{\bH[k]\right\}_{k=1}^K$ are calculated, and their dominant right singular vectors are determined $\left\{\tilde{\bV}[k]\right\}_{k=1}^K$. Note that each element of $\cH$ (and $\cV$) is a set of $K$ matrices for the $K$ subcarriers.}
\item{Nearest neighbor partitioning:}{ In this step, the points in $\cV$ are partitioned into $N_\mathrm{CB}$ Voronoi cells with respect to the codewords in $\cF_\mathrm{U}$ to minimize the average distortion. To do that, we first define the quantization map $\bC\left(\left\{\tilde{\bV}[k]\right\}_{k=1}^K\right)$, that determines the closest codeword in $\cF_\mathrm{U}$ to $\left\{\tilde{\bV}[k]\right\}_{k=1}^K$ in terms of the average squared chordal distance $\Phi_\mathrm{chord}(.)$, as
\begin{equation}
\bC\left(\left\{\tilde{\bV}[k]\right\}_{k=1}^K\right)=\arg\min_{\bX \in \cF_\mathrm{U}} \Phi_\mathrm{chord}\left(\bX,\left\{\tilde{\bV}[k]\right\}_{k=1}^K\right).
\label{eq:centroid}
\end{equation}

Once the codeword closest to each point in $\cV$ is determined, these points can be partitioned into $N_\mathrm{CB}$ sets $\cR_n, n=1,2,..., N_\mathrm{CB}$ as follows
\begin{equation}
\cR_n=\left\{\left\{\tilde{\bV}[k]\right\}_{k=1}^K \in \cV \left| \bC\left(\left\{\tilde{\bV}[k]\right\}_{k=1}^K\right) = \bF_n^\mathrm{U} \right.\right\}.
\label{eq:Cell}
\end{equation}
}
\item{Centroid calculation:}{ The centroid of each partition $\cR_n$ is then derived to minimize the average distortion for this partition. Hence, the objective of this step is to calculate the new codeword $\bF_n^\mathrm{U}$ that solves}
\begin{equation}
\bF_n^\mathrm{U}=\arg\min_{\bX \in \cU_{N_\mathrm{BS} \times N_\mathrm{S}}} \bbE \left[\Phi_\mathrm{chord}\left(\bX, \left\{\tilde{\bV}[k]\right\}_{k=1}^K\right)\left|\left\{ \tilde{\bV}[k]\right\}_{k=1}^K \in \cR_n \right.\right].
\label{eq:cent_avg}
\end{equation}

Minimizing the objective function of the problem in \eqref{eq:cent_avg} is similar to minimizing the function $\Phi_\mathrm{chord}\left(\bX, \left\{\tilde{\bV}[k]\right\}_{k=1}^K\right)$, whose solution is found to be given by the Karcher mean \cite{Mondal2007,Pitaval2014}. Therefore, the new centroid of \eqref{eq:cent_avg} can be calculated in a closed form as
\begin{equation}
\bM_n=\mathrm{eig}_{1:N_\mathrm{S}}\left(\sum_{\cR_n} \sum_{k=1}^K \tilde{\bV}[k] \tilde{\bV}^*[k]\right),
\label{eq:cent_eig}
\end{equation}
where $\mathrm{eig}_{1:N_\mathrm{S}}(\bX)$ represents the first $N_\mathrm{S}$ eigenvectors of the matrix $\bX$ corresponding to the $N_\mathrm{S}$ largest eigenvalues.

\item{RF codewords approximation:}{ The final objective of Algorithm \ref{alg:RF_CB} is to construct an RF codebook $\cF_\mathrm{RF}$ that minimizes the distortion in \eqref{eq:Dist}. Using the triangle inequality on the chordal distances \cite{Pitaval2011}, the additional distortion due to the RF hardware constraints can be bounded by}
\begin{equation}
\Phi_\mathrm{chord}\left(\bU^\mathrm{RF}_n, \left\{\tilde{\bV}[k]\right\}_{k=1}^K\right)-\Phi_\mathrm{chord}\left(\bF^\mathrm{U}_n, \left\{\tilde{\bV}[k]\right\}_{k=1}^K\right) \leq d^2_\mathrm{chord}\left(\bF^\mathrm{U}_n, \bU^\mathrm{RF}_n\right),
\label{eq:RF_Dist_bound}
\end{equation}
where $\bU_n^\mathrm{RF}$ is the $N_\mathrm{S}$ dominant left singular vectors of the $n$th RF codeword. As the chordal distance between two Grassmannian points $\bX$,$\bY$ $\in \cU_{N_\mathrm{BS}\times N_\mathrm{S}}$ is invariant to the right multiplication of any of them by a unitary matrix in $\cU_{N_\mathrm{S} \times N_\mathrm{S}}$, then we have $d^2_\mathrm{chord}\left(\bF^\mathrm{U}_n, \bU^\mathrm{RF}_n\right)=d^2_\mathrm{chord}\left(\bF^\mathrm{U}_n, \bU^\mathrm{RF}_n \bV^\mathrm{RF}_n\right)=d^2_\mathrm{chord}\left(\bF^\mathrm{U}_n, \bF_n^\mathrm{RF} \left(\bF_n^{\mathrm{RF}^*} \bF_n^\mathrm{RF}\right)^{-\frac{1}{2}}\right)$. So, our objective is to solve
\begin{equation}
\bF_n^{\mathrm{RF}^\star}=\arg\min_{|\bF_n^\mathrm{RF}|_{p,q}=1}d^2_\mathrm{chord}\left(\bF^\mathrm{U}_n, \bF_n^\mathrm{RF} \left(\bF_n^{\mathrm{RF}^*} \bF_n^\mathrm{RF}\right)^{-\frac{1}{2}}\right).
\label{eq:RF_prob}
\end{equation}

Finding the exact solution of \eqref{eq:RF_prob} is non-trivial because of the constant-modulus constraint on the entries of $\bF_n^\mathrm{RF}$. For the sake of a closed-form approximated solution, however, we make the following two approximations, that will be shown by simulations in \sref{sec:Results} to give very good results compared with the optimal unconstrained solution. (i) For large mmWave MIMO channels, the columns of $\bF^\mathrm{RF}_n$ can be chosen to be nearly orthogonal, i.e., $\left(\bF_n^{\mathrm{RF}^*} \bF_n^\mathrm{RF}\right) \approx \bI$. (ii) The hybrid precoding and the unconstrained points $\bF_n^\mathrm{RF} \left(\bF_n^{\mathrm{RF}^*} \bF_n^\mathrm{RF}\right)^{-\frac{1}{2}}$ and $\bF_n^\mathrm{U}$ can be made very close \cite{ElAyach2014}. Hence, by leveraging the locally Euclidean property of the Grassmann manifold, the chordal distance in \eqref{eq:RF_prob} can be replaced by the Euclidean distance \cite{ElAyach2014,Lee2012}. Therefore, minimizing the distortion in \eqref{eq:RF_prob} is approximately equal to the following problem
\begin{equation}
\bF_n^{\mathrm{RF}^\star}=\arg\min_{\left|\bF_n^\mathrm{RF}\right|_{p,q}=1}\left\|\bF_n^\mathrm{U}-\bF_n^\mathrm{RF}\right\|_F^2.
\label{eq:RF_approx}
\end{equation}

The problem in \eqref{eq:RF_approx} is a per-entry optimization problem, of which the optimal solution is given by
\begin{equation}
\left[\bF_n^{\mathrm{RF}^\star}\right]_{p,q}= e^{j \measuredangle\left(\left[\bF_n^\mathrm{U}\right]_{p,q}\right)},
\label{eq:RF_approx_ele}
\end{equation}
where the angle $\measuredangle\left(\left[\bF_n^\mathrm{U}\right]_{p,q}\right)$ can then be approximated to the closest quantized angle of the available phase shifters.
\end{itemize}

The convergence of Algorithm \ref{alg:RF_CB} is shown in \figref{fig:Fig1} for a $32 \times 16$ mmWave system with $N_\mathrm{RF}=3$ RF chains, and for a codebook size $128$. The monotonic convergence of Algorithm \ref{alg:RF_CB} to a local optimal solution is guaranteed as the precoding codewords are updated in each iteration according to the nearest neighbor and centroid steps \eqref{eq:centroid} - \eqref{eq:cent_avg} to make an additional reduction in the distortion function \cite{Lloyd1982}. In the next subsection, we extend the developed codebook to the case when $N_\mathrm{S}<N_\mathrm{RF}$.
% A note about the convergence -- of the RF as well
% A note about the quantized phase shifters case
\begin{figure}[t]
\centerline{
\includegraphics[scale=.55]{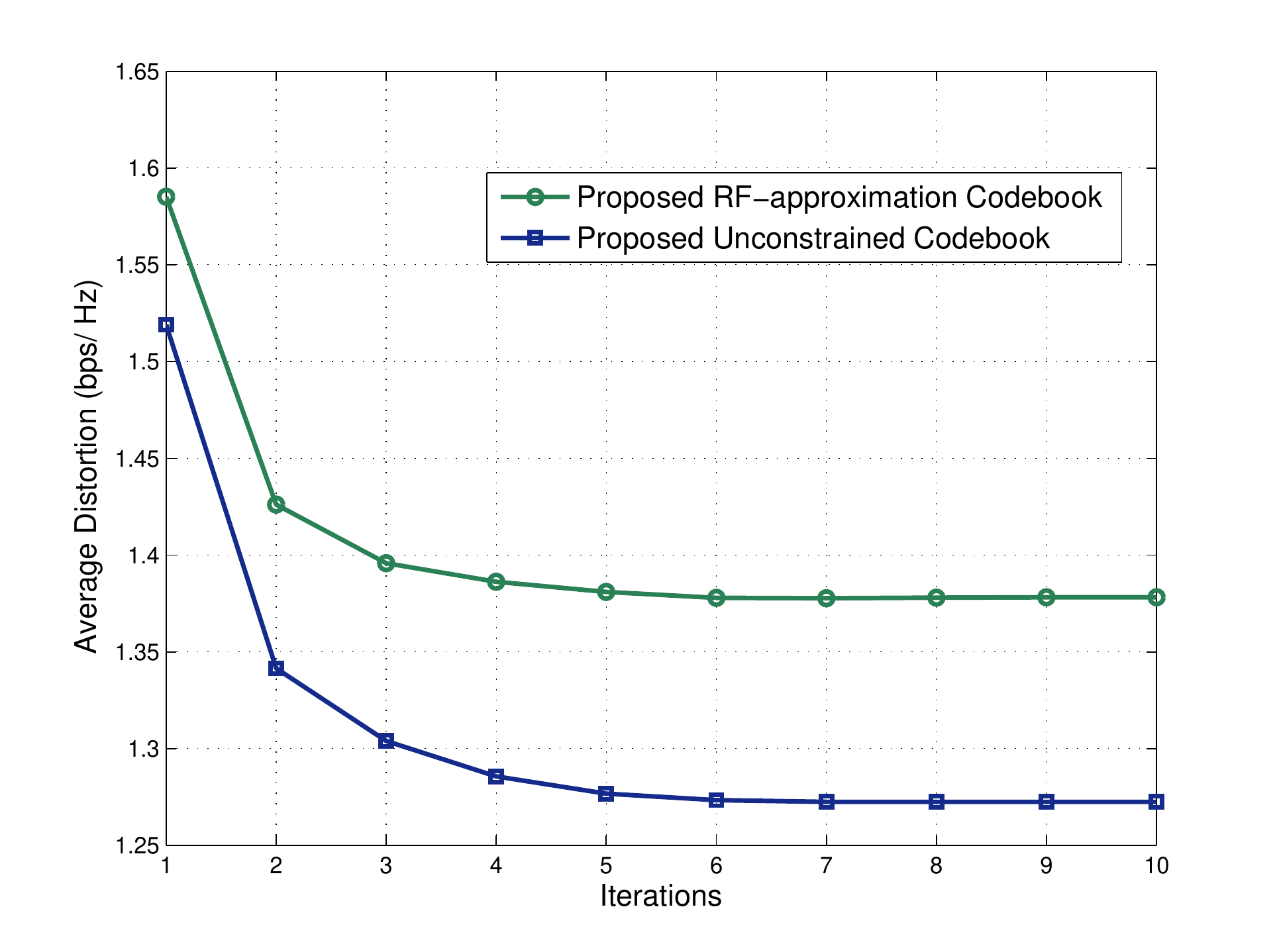}
}
\caption{Average distortion in \eqref{eq:Dist} of the proposed RF codebook constructed using Algorithm \ref{alg:RF_CB} with $N_\mathrm{S}=N_\mathrm{RF}=3$, and codebook size $N_\mathrm{CB}=128$. The rest of the channel and simulation parameters are similar to \figref{fig:Fig2a} described in \sref{sec:Results}. The figure shows the convergence of the unconstrained and RF approximated codebooks to small distortion values.}
\label{fig:Fig1}
\end{figure}
%-----------------------------------------------------------
\subsection{Case 2: $N_\mathrm{S}<N_\mathrm{RF}$} \label{subsec:Ns}
%------------------------------------------------------------
%-----
When $N_\mathrm{S}<N_\mathrm{RF}$, we can see from \eqref{eq:Opt_MI_UP} that the optimal hybrid precoding based mutual information depends on the value of the equivalent baseband precoders, and are not invariant with respect to them because they will not have a unitary structure as when $N_\mathrm{S}=N_\mathrm{RF}$. Hence, both the RF and baseband  precoders $\bF_\mathrm{RF}, \left\{\bF[k]\right\}_{k=1}^K$, need to be quantized and fed back to the transmitter in this case. Inspired by the optimal structure of the baseband precoders in \eqref{eq:Opt_BB_UP} and by Remark 1, we propose to quantize the equivalent baseband precoders $\left\{\bG[k]\right\}$ instead of the baseband precoders. In addition to the intuitive good performance expected to be achieved with equivalent baseband quantization thanks to following the optimal precoders structure, one main advantage of equivalent baseband precoders quantization appears in the favorable structure of the optimal equivalent baseband codebooks as will be discussed shortly. With RF and equivalent baseband precoders quantization, the optimal mutual information is given by
\begin{equation}
\begin{aligned}
\cI^{\star}_\mathrm{HP} & = \\ 
& \hspace{-15pt} \underset{\bF_\mathrm{RF}, \left\{\bF[k]\right\}_{ k=1}^K}  \max
& &\frac{1}{K} \sum_{k=1}^{K} \log_2 \left|\bI_{N_\mathrm{MS}}+\frac{\rho}{N_\mathrm{S}} \bH[k] \bF_\mathrm{RF} \left(\bF_\mathrm{RF}^* \bF_\mathrm{RF}\right)^{-\frac{1}{2}}\bG[k] \bG^*[k] \left(\bF_\mathrm{RF}^* \bF_\mathrm{RF}\right)^{-\frac{1}{2}} \bF_\mathrm{RF}^*  \bH^*[k] \right| \\
& \hspace{5pt} \text{s.t.}
& &  \bF_\mathrm{RF } \in \cF_\mathrm{RF}, \\
&&& \bG[k] \in \mathcal{G}_\mathrm{BB} \subseteq \cU_{N_\mathrm{RF} \times N_\mathrm{S}}, \ \ k=1, 2, ..., K, \label{eq:Opt_Feedback_G}
\end{aligned}
\end{equation}
where the constraint $\mathcal{G}_\mathrm{BB} \subseteq \cU_{N_\mathrm{RF} \times N_\mathrm{S}}$ on the equivalent baseband precoders codebook $\mathcal{G}_\mathrm{BB}$ follows from the unitary power constraint on the hybrid precoders, which requires the equivalent baseband precoders to have a unitary structure. Before delving into the design of RF precoders codebook $\cF_\mathrm{RF}$ and the equivalent baseband precoders codebook $\cG_\mathrm{BB}$, we make the following remark on the codebook structure of the optimal equivalent baseband precoders. \\
\hspace{-15pt}\textbf{Remark 3.} \textit{Regardless of the RF codebook, the optimal codebook for the equivalent baseband precoders $\left\{\bG[k]\right\}_{k=1}^K$ under a unitary hybrid precoding constraint is unitary.}

In the remaining part of this subsection, we present the proposed design and construction of the RF and equivalent baseband precoders codebooks, $\cF_\mathrm{RF}$ and $\cG_\mathrm{BB}$.

\textbf{Hybrid Codebook Design Criterion:}
The objective now is to design $\cF_\mathrm{RF}$ and $\cG_\mathrm{BB}$ to minimize the distortion function $\cD\left(\cF_\mathrm{RF}, \cG_\mathrm{BB}\right)$ defined as the average mutual information loss of limited feedback frequency selective hybrid precoding compared with the unconstrained perfect channel knowledge solution. Formally, the distortion function $\cD\left(\cF_\mathrm{RF}, \cG_\mathrm{BB}\right)$ is written as
\begin{equation}
\cD\left(\cF_\mathrm{RF}, \cG_\mathrm{BB}\right) = \bbE_{\left\{\bH[k]\right\}_{k=1}^K}\left[\cI_\mathrm{UC}^\star-\cI^\star_\mathrm{HP}\right],
\label{eq:Dist_HP}
\end{equation}
where $\cI_\mathrm{UC}^\star$ and $\cI_\mathrm{HP}^\star$ are as defined in \eqref{eq:Opt_MI_UC} and \eqref{eq:Opt_Feedback_G}, respectively.

The main challenge of this distortion function is that the hybrid precoding mutual information depends on the joint RF and equivalent baseband precoders codebooks as shown in \eqref{eq:Opt_Feedback_G}, which makes the  direct design of these codebooks to minimize the distortion in \eqref{eq:Dist_HP} non-trivial. Next, we leverage the optimal baseband prcoders structure in \sref{sec:Optimal} to derive an upper bound on the limited feedback hybrid precoding distortion in \eqref{eq:Dist_HP}. This bound will attempt to decouple the distortion impact of the RF and equivalent baseband precoding codebooks, and therefore simplify the hybrid codebook design problem. The limited feedback hybrid precoding distortion $\cD\left(\cF_\mathrm{RF}, \cG_\mathrm{BB}\right)$ in \eqref{eq:Dist_HP} can be written as
\begin{align}
\cD\left(\cF_\mathrm{RF}, \cG_\mathrm{BB}\right) &= \bbE_{\left\{\bH[k]\right\}_{k=1}^K}\left[\cI_\mathrm{UC}^\star-\cI^\star_\mathrm{HP}\right], \\
&\begin{aligned} \stackrel{(a)}{=} &\bbE_{\left\{\bH[k]\right\}_{k=1}^K} \underbrace{\left[\cI_\mathrm{UC}^\star-\underset{\bF_\mathrm{RF} \in \cF_\mathrm{RF}}  \max \frac{1}{K} \sum_{k=1}^{K} \log_2 \left|\bI_{N_\mathrm{S}}+\frac{\rho}{N_\mathrm{S}} \left[\overline{\boldsymbol{\Sigma}}[k]\right]_{1:N_\mathrm{S}, 1:N_\mathrm{S}}^2  \ \right|\right] }_{\Delta\cI_\mathrm{RF}}\\
& + \bbE_{\left\{\bH[k]\right\}_{k=1}^K}  \underbrace{\left[\underset{\bF_\mathrm{RF} \in \cF_\mathrm{RF}}  \max \frac{1}{K} \sum_{k=1}^{K} \log_2 \left|\bI_{N_\mathrm{S}}+\frac{\rho}{N_\mathrm{S}} \left[\overline{\boldsymbol{\Sigma}}[k]\right]_{1:N_\mathrm{S}, 1:N_\mathrm{S}}^2  \ \right|- \cI^\star_\mathrm{HP}\right]}_{\Delta\cI_{BB|RF}}\end{aligned},\label{eq:Dist_HP_Terms}\\
& = \cD\left(\cF_\mathrm{RF}\right)+\cD\left(\cG_\mathrm{BB}\left|\cF_\mathrm{RF}\right.\right) \label{eq:Dist_sum},
\end{align}
where (a) follows by adding and subtracting the optimal hybrid precoding based mutual information with optimal equivalent baseband precoding knowledge in \eqref{eq:Opt_MI_UP}. The first term is therefore the average mutual information loss due to RF codebook alone, $\cD\left(\cF_\mathrm{RF}\right)$, while the second term represents the additional loss with equivalent baseband precoders quantization $\cD\left(\cG_\mathrm{BB}\left|\cF_\mathrm{RF}\right.\right)$. Exploiting the optimal baseband precoders design in \eqref{eq:Opt_BB_UP}, we can bound mutual information loss due to RF quantization as
\begin{align}
\Delta\cI_\mathrm{RF}   & \stackrel{(a)}{=} \frac{1}{K} \sum_{k=1}^K \log_2 \left|\bI+\frac{\rho}{N_\mathrm{S}}\tilde{\boldsymbol{\Sigma}}[k]^2 \right| \nonumber\\
& - \underset{\bF_\mathrm{RF} \in \cF_\mathrm{RF}}  \max  \frac{1}{K} \sum_{k=1}^K \sum_{i=1}^{N_\mathrm{S}}\log_2 \left(1+\frac{\rho}{N_\mathrm{S}} \lambda_i\left(\boldsymbol{\Sigma}[k] \bV^*[k] \bU_\mathrm{RF} \bV_\mathrm{RF}^* \bV_\mathrm{RF} \bU_\mathrm{RF}^* \bV[k] \boldsymbol{\Sigma}^*[k]  \right)\right),\\
& \stackrel{(b)}{\leq} \frac{1}{K} \sum_{k=1}^K \log_2 \left|\bI+\frac{\rho}{N_\mathrm{S}}\tilde{\boldsymbol{\Sigma}}[k]^2 \right| \nonumber\\
& - \underset{\bF_\mathrm{RF} \in \cF_\mathrm{RF}}  \max  \frac{1}{K} \sum_{k=1}^K \sum_{i=1}^{N_\mathrm{S}}\log_2 \left(1+\frac{\rho}{N_\mathrm{S}} \lambda_i\left(\tilde{\boldsymbol{\Sigma}}[k] \tilde{\bV}^*[k] \bU_\mathrm{RF}  \bU_\mathrm{RF}^* \tilde{\bV}[k] \tilde{\boldsymbol{\Sigma}}^*[k]  \right)\right),\\
&{=} \frac{1}{K} \sum_{k=1}^K \log_2 \left|\bI+\frac{\rho}{N_\mathrm{S}}\tilde{\boldsymbol{\Sigma}}[k]^2 \right| \nonumber \\
& - \underset{\bF_\mathrm{RF} \in \cF_\mathrm{RF}}  \max \frac{1}{K} \sum_{k=1}^K \log_2 \left|\bI+\frac{\rho}{N_\mathrm{S}} \tilde{\boldsymbol{\Sigma}}[k] \tilde{\bV}^*[k] \bU_\mathrm{RF} \bU_\mathrm{RF}^* \tilde{\bV}[k] \tilde{\boldsymbol{\Sigma}}^*[k] \right| ,\\
&\stackrel{(c)}{\approx} \underset{\bF_\mathrm{RF} \in \cF_\mathrm{RF}}{\min}  \frac{1}{K} \sum_{k=1}^K \overline{d}^2_\mathrm{chord}\left(\bU_\mathrm{RF}, \tilde{\bV}[k]\right),\\
& = \underset{\bF_\mathrm{RF} \in \cF_\mathrm{RF}}{\min}  \overline{\Phi}_\mathrm{chord}\left(\bU_\mathrm{RF}, \left\{\tilde{\bV}[k]\right\}_{k=1}^K\right), \label{eq:MI_Loss_RF}
\end{align}
where (a) follows from the design of the optimal baseband precoder in \eqref{eq:Opt_BB_UP}. The bound in (b) follows by considering only the $N_\mathrm{S}$ dominant right singular vectors of the channel, i.e., the first $N_\mathrm{S}$ columns of $\bV[k]$, and (c) follows by considering the large mmWave MIMO approximations used in \eqref{eq:Large_Approx}. In (c), $\overline{d}_\mathrm{chord}$ is the generalized chordal distance between subspaces of different dimensions defined as $\overline{d}^2_\mathrm{chord}\left(\bU_\mathrm{RF}, \tilde{\bV}[k]\right)=\min\left(N_\mathrm{RF}, N_\mathrm{S}\right) - \left\|\bU_\mathrm{RF}^* \tilde{\bV}[k]\right\|_F^2$ \cite{Ye2014}, where the dimensions of $\bU_\mathrm{RF}$ and $\tilde{\bV}[k]$ are $N_\mathrm{BS} \times N_\mathrm{RF}$ and $N_\mathrm{BS} \times N_\mathrm{S}$, respectively. Finally, $\overline{\Phi}_\mathrm{chord}(.)$ is defined as in \eqref{eq:Dist}, but with respect to the generalized chordal distance $\overline{d}_\mathrm{chord}(.)$. Given the result in \eqref{eq:MI_Loss_RF}, we reach the following bound on $\cD\left(\cF_\mathrm{RF}\right)$
\begin{equation}
\cD\left(\cF_\mathrm{RF}\right) \leq \bbE\left[ \min_{\bF_\mathrm{RF} \in \cF_\mathrm{RF}} \frac{1}{K} \sum_{k=1}^K \overline{\Phi}_\mathrm{chord}\left(\bF_\mathrm{RF}, \left\{\tilde{\bV}[k]\right\}_{k=1}^K\right)\right].
\label{eq:Dist_UP_RF}
\end{equation}

Now, we derive a similar bound on the additional distortion due to the equivalent baseband quantization given a certain RF codebook $\cD\left(\cG_\mathrm{BB} \left|\cF_\mathrm{RF}\right.\right)$. Let $\bF_\mathrm{RF}^\star \in \cF_\mathrm{RF}$ be the solution of \eqref{eq:Opt_MI_UP}, i.e., the solution of the first term in $\Delta\cI_{BB|RF}$ in \eqref{eq:Dist_HP_Terms}, and $\overline{\boldsymbol{\Sigma}}^\star[k]$ be the corresponding $\overline{\boldsymbol{\Sigma}}[k]$. As $\bF_\mathrm{RF}^\star$ represents a feasible (not necessarily the optimal) solution of the problem in \eqref{eq:Opt_Feedback_G}, then $\cD\left(\cG_\mathrm{BB} \left|\cF_\mathrm{RF}\right.\right)$ in \eqref{eq:Dist_sum} can be bounded as
\begin{align}
\cD\left(\cG_\mathrm{BB} \left| \cF_\mathrm{RF} \right. \right) &\leq \bbE\left[ \frac{1}{K}  \sum_{k=1}^{K} \left( \log_2 \left|\bI_{N_\mathrm{S}}+\frac{\rho}{N_\mathrm{S}} \left[\overline{\boldsymbol{\Sigma}}[k]\right]_{1:N_\mathrm{S}, 1:N_\mathrm{S}}^2  \right|\right. \right. \nonumber \\
&+ \left. \left.\underset{\bG[k] \in \cG_\mathrm{BB}}  \max  \log_2 \left|\bI_{N_\mathrm{S}}+\frac{\rho}{N_\mathrm{S}} \left[\overline{\boldsymbol{\Sigma}}[k]\right]_{1:N_\mathrm{S}, 1:N_\mathrm{S}}^2 \left[\overline{\bV}[k]\right]_{:, 1:N_\mathrm{S}}^* \bG[k] \bG^*[k] \left[\overline{\bV}[k]\right]_{:, 1:N_\mathrm{S}} \right| \right)\right],\\
&\stackrel{(a)}{\approx} \bbE \left[ \underset{\bG \in \cG_\mathrm{BB}}{\min}  \frac{1}{K} \sum_{k=1}^K {d}_\mathrm{chord}\left(\bG, \left[\overline{\bV}[k]\right]_{:, 1:N_\mathrm{S}}\right) \right],\\
& =\bbE\left[ \min_{\bG \in \cG_\mathrm{BB}} {\Phi}_\mathrm{chord}\left(\bG, \left\{\left[\overline{\bV}[k]\right]_{:,1:N_\mathrm{S}}\right\}_{k=1}^K\right)\right], \label{eq:Dist_UP_BB}
\end{align}
where (a) follows by considering the large mmWave MIMO approximations used in \eqref{eq:Large_Approx}. The codebook design objective is then to minimize the upper bound on the distortion function, that is given by the bounds in \eqref{eq:Dist_UP_BB}, \eqref{eq:Dist_UP_RF} on $\cD\left(\cF_\mathrm{RF}\right)+\cD\left(\cG_\mathrm{BB}\left|\cF_\mathrm{RF}\right.\right)$.

\textbf{Hybrid Codebook Construction:}
Given the distortion function upper bounds in \eqref{eq:Dist_UP_RF} and \eqref{eq:Dist_UP_BB}, we will design the RF codebook $\cF_\mathrm{RF}$ to minimize the derived bound on $\cD\left(\cF_\mathrm{RF}\right)$. Then, we will design the equivalent baseband codebook $\cG_\mathrm{BB}$ to minimize the bound on the additional distortion of the equivalent baseband precoders quantization $\cD\left(\cG_\mathrm{BB}\left|\cF_\mathrm{RF}\right.\right)$. As the distortion bounds in \eqref{eq:Dist_UP_RF} and \eqref{eq:Dist_UP_BB} are similar to the expression of the RF codebook distortion in \eqref{eq:Dist}, we  use Algorithm \ref{alg:RF_CB} to design the hybrid RF and equivalent baseband codebooks $\cF_\mathrm{RF}$ and $\cG_\mathrm{BB}$. For the RF codebook, Algorithm \ref{alg:RF_CB} will be used, but with replacing the chordal distance $d_\mathrm{chord}\left(.\right)$ in \eqref{eq:centroid}, \eqref{eq:cent_avg}, \eqref{eq:RF_Dist_bound}, \eqref{eq:RF_prob} by the generalized chordal distance between subspaces of different dimensions in \eqref{eq:MI_Loss_RF}. To build the unitary equivalent baseband precoders codebook, Algorithm \ref{alg:RF_CB} will be also used, but without step 4-c as no RF approximation is required. Even though the dependence of the distortion function $\cD\left(\cG_\mathrm{BB} \left|\cF_\mathrm{RF}\right.\right)$ on the RF codebook $\cF_\mathrm{RF}$ is relaxed in the design, i.e., the RF and baseband codebooks are sequentially designed, the developed hybrid codebooks achieve  good performance compared with the perfect channel knowledge case as will be shown in \sref{sec:Results}.

%%%%%%%%%%%%%%%%%%%%%%%%%%%%%%%%%%%%%%%%%%%%%%%%%%%%%%%%%%%%%%%%%%%%%%%%%%%%%%%%%%%%%%%%%%%%%%%%%%%%%%%%%%%%%
\section{Gram-Schmidt Based Greedy Hybrid Precoding} \label{sec:Greedy}
%%%%%%%%%%%%%%%%%%%%%%%%%%%%%%%%%%%%%%%%%%%%%%%%%%%%%%%%%%%%%%%%%%%%%%%%%%%%%%%%%%%%%%%%%%%%%%%%%%%%%%%%%%%%%

The optimal hybrid precoding design for any given RF codebook was derived in \sref{sec:Optimal}. An exhaustive search over the RF codebook, however, is still required to find the optimal RF precoder in \eqref{eq:Opt_MI_UP}. This search may be of high complexity, especially for large antenna systems. Therefore, and inspired by the optimal baseband precoder structure in \eqref{eq:Opt_BB_UP}, we develop a greedy frequency selective hybrid precoding algorithm in this section based on Gram-Schmidt orthogonalization. Different from prior work that mainly depends on heuristic ideas for the joint design of the RF and baseband precoders \cite{ElAyach2014,Alkhateeb2014,Chen2015}, we will make statements on the optimality of the proposed algorithm in some cases, even though it sequentially designs the RF and baseband precoders.

Equation \eqref{eq:Opt_MI_UP} showed that the optimal hybrid precoding based mutual information can be written as a function of the RF precoding alone. Hence, the remaining problem was to determine the best RF precoding matrix, i.e., the best $N_\mathrm{RF}$ beamforming vectors, from the RF precoding codebook $\cF_\mathrm{RF}$. This requires making an exhaustive search over the matrices codewords in $\cF_\mathrm{RF}$. A natural greedy approach to construct the hybrid precoder is to iteratively select the $N_\mathrm{RF}$ RF beamforming vectors to maximize the mutual information. In this paper, we call this  the direct greedy hybrid precoding (DG-HP) algorithm. For simplicity of exposition, we will assume that the RF beamforming vectors of the $N_\mathrm{RF}$ RF chains are to be selected from the same vector codebook $\cF^\mathrm{v}_\mathrm{RF}=\left\{\bff^\mathrm{RF}_1, ..., \bff^\mathrm{RF}_{N_\mathrm{CB}}\right\}$, but choosing unique codewords. Extensions to the case when each of the $N_\mathrm{RF}$ RF beamforming vectors is taken from a different codebook is straightforward.

The operation of the DG-HP algorithm then consists of  $N_\mathrm{RF}$ iterations. In each iteration, the RF beamforming vector from $\cF_\mathrm{RF}^\mathrm{v}$ that maximizes the mutual information at this iteration will be selected. Let the $N_\mathrm{BS}\times (i-1)$ matrix $\bF_\mathrm{RF}^{(i-1)}$ denote the RF precoding matrix at the end of the $(i-1)$th iteration. Then by leveraging the optimal baseband precoder structure in \eqref{eq:Opt_BB_UP}, the objective of the $i$th iteration is to select $\bff^\mathrm{RF}_{n} \in \cF_\mathrm{RF}^\mathrm{v}$ that solves
\begin{equation}
\cI_\mathrm{HP}^{(i)}=\underset{\bff^\mathrm{RF}_n \in \cF_\mathrm{RF}^\mathrm{v}}{\max}\frac{1}{K} \sum_{k=1}^K \sum_{\ell=1}^{i} \log_2\left(1 + \frac{\rho}{N_\mathrm{S}} \lambda_\ell\left(\bH\left[k\right] \dot{\bF}_\mathrm{RF}^{(i,n)}\left({\dot{\bF}_\mathrm{RF}^{{(i,n)}^*}}\dot{\bF}_\mathrm{RF}^{(i,n)}\right)^{-1} {\dot{\bF}_\mathrm{RF}^{{(i,n)}^*}} \bH\left[k\right]^* \right)\right),
\label{eq:DG_HP}
\end{equation}
with $\dot{\bF}_\mathrm{RF}^{(i,n)}=\left[\bF_\mathrm{RF}^{(i-1)}, \bff^\mathrm{RF}_n\right]$. The best vector $\bff^\mathrm{RF}_{n^\star}$ will be then added to the RF precoding matrix to form $\bF_\mathrm{RF}^{(i)}=\left[\bF_\mathrm{RF}^{(i-1)},  \bff^\mathrm{RF}_{n^\star}\right]$. The achievable mutual information with this algorithm is then $\cI_\mathrm{HP}^\mathrm{DG-HP}=\cI_\mathrm{HP}^{(N_\mathrm{RF})}$. The main limitation of this algorithm is that it still requires an exhaustive search over $\cF_\mathrm{RF}^\mathrm{v}$ and eigenvalues calculation in each iteration. The objective of this section is to develop a low-complexity algorithm that has a similar (or very close) performance to this DG-HP algorithm. In the next subsection, we will make the first step towards this goal by proving that a Gram-Schmidt based algorithm can lead to exactly the same performance of the DG-HP. This will be then leveraged into the design of our algorithm in \sref{subsec:Approx_GS}.

\subsection{Gram-Schmidt Based Greedy Hybrid Precoding} \label{subsec:GS_Main}

In hybrid analog/digital precoding architectures, the effective channel seen at the baseband is through the RF precoders lens. This gives the intuition that it is better for the RF beamforming vectors to be orthogonal (or close to orthogonal), as this physically means that the effective channel will have a better coverage over the dominant subspaces belonging to the actual channel matrix. This intuition is also confirmed by the structure of the optimal baseband precoder discussed in Remark 1, as the overall matrix $\bF_\mathrm{RF}\left(\bF_\mathrm{RF}^* \bF_\mathrm{RF}\right)^{-\frac{1}{2}}$ has a semi-unitary structure. Indeed, this observation can also be related to the structure of the solutions of the nearest matrix and nearest tight frame problems \cite{Choi2006,Tropp2005}. This note means that in each iteration $i$ of the greedy hybrid precoding algorithm in \eqref{eq:DG_HP} with a selected codeword $\bff^\mathrm{RF}_{n^\star}$, the additional mutual information gain over the previous iterations is due to the contribution of the component of $\bff^\mathrm{RF}_{n^\star}$ that is orthogonal on the existing RF precoding matrix $\bF_\mathrm{RF}^{(i-1)}$. Based on that, we modify the DG-HP algorithm by adding a Gram-Schmidt orthogonalization step in each iteration $i$ to project the candidate beamforming codewords on the orthogonal complement of the subspace spanned by the selected codewords in $\bF_\mathrm{RF}^{(i-1)}$. This can be simply done by multiplying the candidate vectors by the projection matrix ${\bP^{(i-1)}}^{\perp}=\left(\bI_{N_\mathrm{BS}}-\bF_\mathrm{RF}^{(i-1)}\left(\bF_\mathrm{RF}^{(i-1)^*} \bF_\mathrm{RF}^{(i-1)}\right)^{-1} \bF_\mathrm{RF}^{(i-1)^*}\right)$. Given the optimal precoder design in \eqref{eq:Opt_BB_UP}, the mutual information at the $i$th iteration of the modified Gram-Schmidt hybrid precoding (GS-HP) algorithm can be written as
\begin{align}
\overline{\cI}_\mathrm{HP}^{(i)} &=\underset{\bff^\mathrm{RF}_n \in {\cF}_\mathrm{RF}^{\mathrm{v}}}{\max}\frac{1}{K} \sum_{k=1}^K \sum_{\ell=1}^{i} \log_2\left(1 + \frac{\rho}{N_\mathrm{S}} \lambda_\ell\left(\bH\left[k\right] \overline{\bF}_\mathrm{RF}^{(i,n)}\left({\overline{\bF}_\mathrm{RF}^{(i,n)}}^* \overline{\bF}_\mathrm{RF}^{(i,n)}\right)^{-1} {\overline{\bF}_\mathrm{RF}^{(i,n)}}^* \bH\left[k\right]^* \right)\right),\label{eq:GS_HP}\\
&\stackrel{(a)}{=} \underset{\bff^\mathrm{RF}_n \in {\cF}_\mathrm{RF}^{\mathrm{v}}}{\max}\frac{1}{K} \sum_{k=1}^K \sum_{\ell=1}^{i} \log_2\left(1 + \frac{\rho}{N_\mathrm{S}} \lambda_\ell\left( \bT^{(i-1)}+ \bH[k] {\bP^{(i-1)}}^{\perp} \bff^\mathrm{RF}_n  {\bff^\mathrm{RF}_n}^* {{\bP^{(i-1)}}^{\perp}}^* \bH^*[k] \right)\right),
\label{eq:GS_HP2}
\end{align}
with $\overline{\bF}_\mathrm{RF}^{(i,n)}=\left[\bF_\mathrm{RF}^{(i-1)}, {\bP^{(i-1)}}^{\perp} \bff^\mathrm{RF}_n\right]$, and $\bT^{(i-1)}=\bH\left[k\right] \bF_\mathrm{RF}^{(i-1)}\left({\bF_\mathrm{RF}^{(i-1)}}^* \bF_\mathrm{RF}^{(i-1)}\right)^{-1} {\bF_\mathrm{RF}^{(i-1)}}^* \bH\left[k\right]^*$. Note that $\bT^{(i-1)}$ is a constant matrix at iteration $i$, and (a) follows from the Gram-Schmidt orthogonalization which allows the matrix $\overline{\bF}_\mathrm{RF}^{(i,n)}\left({\overline{\bF}_\mathrm{RF}^{(i,n)}}^* \overline{\bF}_\mathrm{RF}^{(i,n)}\right)^{-\frac{1}{2}}$ at iteration $i$ to be written as $\left[\bF_\mathrm{RF}^{(i-1)}\left({\bF_\mathrm{RF}^{(i-1)}}^* \bF_\mathrm{RF}^{(i-1)}\right)^{-\frac{1}{2}},  {\bP^{(i-1)}}^{\perp} \bff_n^\mathrm{RF} \right]$. Hence, the eigenvalues calculation in \eqref{eq:GS_HP2} can be calculated as a rank-1 update of the previous iteration eigenvalues, which reduces the overall complexity \cite{Br2006}. The best vector $\bff^\mathrm{RF}_{n^\star}$ will be then added to the RF precoding matrix to form $\bF_\mathrm{RF}^{(i)}=\left[\bF_\mathrm{RF}^{(i-1)},  \bff^\mathrm{RF}_{n^\star}\right]$. At the end of the $N_\mathrm{RF}$ iterations, the achieved mutual information is $\cI_\mathrm{HP}^\mathrm{GS-HP}=\overline{\cI}_\mathrm{HP}^{(N_\mathrm{RF})}$.  In the following proposition, we prove that this Gram-Schmidt hybrid precoding algorithm is exactly equivalent to the DG-HP algorithm.

\begin{proposition}
The achieved mutual information of the direct greedy hybrid precoding algorithm in \eqref{eq:DG_HP} and the Gram-Schmidt based hybrid precoding algorithm in \eqref{eq:GS_HP} are exactly equal, i.e., $\cI_\mathrm{HP}^\mathrm{DG-HP}=\cI_\mathrm{HP}^\mathrm{GS-HP}$.
\label{prop:GS}
\end{proposition}
\begin{proof}
See Appendix \ref{app:GS}.
\end{proof}

\subsection{Approximate Gram-Schmidt Based Greedy Hybrid Precoding} \label{subsec:Approx_GS}
\begin{algorithm} [!t]                     % enter the algorithm environment
\caption{Approximate Gram-Schmidt Based Frequency Selective Hybrid Precoding}          % give the algorithm a caption
\label{alg:GS_HP}                           % and a label for \ref{} commands later in the document
\begin{algorithmic} %[1]                   % enter the algorithmic environment
    \State \textbf{Initialization}
    \State 1) \begin{varwidth}[t]{\linewidth} Construct $\boldsymbol{\Pi} = \tilde{\bV}_\mathrm{\bH} \tilde{\boldsymbol{\Sigma}}_\mathrm{\bH} $, with $\tilde{\boldsymbol{\Sigma}}_\mathrm{\bH}=\text{diag}\left(\tilde{\boldsymbol{\Sigma}}_\mathrm{1}, ..., \tilde{\boldsymbol{\Sigma}}_\mathrm{K}\right)$ and $\tilde{\bV}_\mathrm{\bH}=\left[\tilde{\bV}_\mathrm{1}, ..., \tilde{\bV}_\mathrm{K}\right]$. \\ Set $\bF_\mathrm{RF}=$  Empty Matrix. Define $\bA_\mathrm{CB}=\left[{\bff}_1^\mathrm{RF}, ...,{\bff}^\mathrm{RF}_{N_\mathrm{CB}^\mathrm{v}} \right]$, where ${\bff}^\mathrm{RF}_n, n=1, ..., {N_\mathrm{CB}^\mathrm{v}}$ are \\ the codewords in ${\cF}_\mathrm{RF}$ \\
    	\end{varwidth}
    \State \textbf{RF Precoder Design}
    \State 2) \text{For}{ $i, i = 1, ..., N_\mathrm{RF}$}
        \State  $\hspace{20pt}$ a) $ \boldsymbol\Psi = \boldsymbol\Pi^*{\bA}_\mathrm{CB}$
        \State  $\hspace{20pt}$ b) $n^\star=\arg\max_{n=1,2,..N_\mathrm{CB}^\mathrm{v}} \left\|\left[\boldsymbol\Psi\right]_{:,n}\right\|_2$.
        \State  $\hspace{20pt}$ c) $\bF_\mathrm{RF}^{(i)} = \left[\bF_\mathrm{RF}^{(i-1)} \bff_{n^\star}^\mathrm{RF}\right]$
        \State $\hspace{20pt}$ d) $\boldsymbol\Pi= \left(\bI_{N_\mathrm{BS}}-\bF_\mathrm{RF}^{(i)}\left(\bF_\mathrm{RF}^{(i)^*} \bF_\mathrm{RF}^{(i)}\right)^{-1}\bF_\mathrm{RF}^{(i)^*}\right) \boldsymbol\Pi$
    \State \textbf{Digital Precoder Design}
    \State 3) $\bF[k]= \bF_\mathrm{RF}^{(N_\mathrm{RF})} \left(\bF_\mathrm{RF}^{(N_\mathrm{RF})^*} \bF_\mathrm{RF}^{(N_\mathrm{RF})} \right)^{-\frac{1}{2}} \left[\overline{\bV}[k]\right]_{:,1:N_\mathrm{S}}, k=1, ..., K$,  with $\overline{\bV}[k]$ defined in \eqref{eq:Opt_BB_UP}
\end{algorithmic}
\end{algorithm}

The main advantage of the Gram-Schmidt hybrid precoding design in \sref{subsec:GS_Main} is that it leads to a near-optimal low-complexity design of the frequency selective hybrid precoding as will be discussed in this section. Given the optimal baseband precoding solution in \eqref{eq:Opt_BB_UP}, the mutual information at the $i$th iteration in \eqref{eq:GS_HP} can be written as
\begin{align}
\overline{\cI}_\mathrm{HP}^{(i)}& = \underset{\bff^\mathrm{RF}_n \in {\cF}_\mathrm{RF}^{\mathrm{v}}}{\max}\frac{1}{K} \sum_{k=1}^K \sum_{\ell=1}^{i} \log_2\left(1 + \frac{\rho}{N_\mathrm{S}} \lambda_\ell\left(\bH\left[k\right] \overline{\bF}_\mathrm{RF}^{(i,n)}\left({\overline{\bF}_\mathrm{RF}^{(i,n)}}^* \overline{\bF}_\mathrm{RF}^{(i,n)}\right)^{-1} {\overline{\bF}_\mathrm{RF}^{(i,n)}}^* \bH\left[k\right]^* \right)\right),\\
&\stackrel{(a)}{\geq} \underset{\bff^\mathrm{RF}_n \in {\cF}_\mathrm{RF}^{\mathrm{v}}}{\max} \frac{1}{K} \sum_{k=1}^K \sum_{\ell=1}^{i} \log_2  \left( 1+\frac{\rho}{N_\mathrm{S}}  \lambda_\ell\left(\tilde{\boldsymbol{\Sigma}}[k] \tilde{\bV}^*[k] \overline{\bF}_\mathrm{RF}^{(i,n)}\left({\overline{\bF}_\mathrm{RF}^{(i,n)}}^* \overline{\bF}_\mathrm{RF}^{(i,n)}\right)^{-1} {\overline{\bF}_\mathrm{RF}^{(i,n)}}^* \tilde{\bV}[k] \tilde{\boldsymbol{\Sigma}}^*[k] \right)\right), \\
& \stackrel{(b)}{\approx} \frac{1}{K} \sum_{k=1}^K \left(\log_2\left|\bI +\frac{\rho}{N_\mathrm{S}}  \tilde{\boldsymbol{\Sigma}}[k]^2 \right| - \text{tr} \left(\tilde{\boldsymbol{\Sigma}}[k]\right) \right)  \nonumber\\  
& \hspace{120pt} + \underset{\bff^\mathrm{RF}_n \in {\cF}_\mathrm{RF}^{\mathrm{v}}}{\max} \frac{1}{K} \sum_{k=1}^K \left\|\tilde{\boldsymbol{\Sigma}}[k] \tilde{\bV}^*[k] \overline{\bF}_\mathrm{RF}^{(i,n)} \left({\overline{\bF}_\mathrm{RF}^{(i,n)}}^* \overline{\bF}_\mathrm{RF}^{(i,n)}\right)^{-\frac{1}{2}} \right\|_F^2,
\end{align}
where the bound in (a) is by considering only the first $N_\mathrm{S}$ dominant singular values of $\bH[k]$, and (b) follows from using the large mmWave MIMO approximations used in \eqref{eq:Large_Approx}. The objective then of the $i$th iteration is to select ${\bff}_n^\mathrm{RF} \in {\cF}_\mathrm{RF}$ that solves
\begin{align}
\bff^\mathrm{RF}_{n^\star}&=\underset{{\bff}^\mathrm{RF}_n \in {\cF}_\mathrm{RF}^\mathrm{v}}{\arg\max}\frac{1}{K} \sum_{k=1}^K \left\|\tilde{\boldsymbol{\Sigma}}[k] \tilde{\bV}^*[k] \overline{\bF}_\mathrm{RF}^{(i,n)} \left({\overline{\bF}_\mathrm{RF}^{(i,n)}}^* \overline{\bF}_\mathrm{RF}^{(i,n)}\right)^{-\frac{1}{2}} \right\|_F^2, \\
& \stackrel{(a)}{=}\left\|\tilde{\boldsymbol{\Sigma}}_\mathrm{\bH} \tilde{\bV}_\mathrm{\bH}^* \bF^{(i-1)}_\mathrm{RF} \left(\bF^{(i-1)^*}_\mathrm{RF} \bF^{(i-1)}_\mathrm{RF}\right)^{-\frac{1}{2}} \right\|_F^2+ \underset{{\bff}^\mathrm{RF}_n \in {\cF}_\mathrm{RF}^\mathrm{v}}{\arg\max}\left\|\tilde{\boldsymbol{\Sigma}}_\mathrm{\bH} \tilde{\bV}_\mathrm{\bH}^* {\bP^{(i-1)}}^\perp {\bff}_n^\mathrm{RF}\right\|_2^2, \label{eq:Opt_GS_Simple}
\end{align}
where $\tilde{\boldsymbol{\Sigma}}_\mathrm{\bH}=\text{diag}\left(\tilde{\boldsymbol{\Sigma}}[1], ..., \tilde{\boldsymbol{\Sigma}}[K]\right)$, $\tilde{\bV}_\mathrm{\bH}=\left[\tilde{\bV}[1], ..., \tilde{\bV}[K]\right]$, and (a) is a result of the Gram-Schmidt processing which makes $\overline{\bF}_\mathrm{RF}^{(i,n)} \left({\overline{\bF}_\mathrm{RF}^{(i,n)}}^* \overline{\bF}_\mathrm{RF}^{(i,n)}\right)^{-\frac{1}{2}}$ with $\overline{\bF}_\mathrm{RF}^{(i,n)}=\left[\bF_\mathrm{RF}^{(i-1)}, {\bP^{(i-1)}}^\perp \bff_n^\mathrm{RF}\right]$ at the $i$th iteration equals to $\left[\bF^{(i-1)}_\mathrm{RF} \left(\bF^{(i-1)^*}_\mathrm{RF} \bF^{(i-1)}_\mathrm{RF}\right)^{-\frac{1}{2}}, {\bP^{(i-1)}}^\perp {\bff}_n^{(i)}\right]$. The problem in \eqref{eq:Opt_GS_Simple} is simple to solve with just a maximum projection step. We call this algorithm the approximate Gram-Schmidt hybrid precoding (Approximate GS-HP) algorithm. As shown in  Algorithm \ref{alg:GS_HP}, the developed algorithm sequentially build the RF and baseband precoding matrices in two separate stages. First, the RF beamforming vectors are iteratively selected to solve \eqref{eq:Opt_GS_Simple}. Then, the baseband precoder is optimally designed according to \eqref{eq:Opt_BB_UP}. Despite its sequential design of the RF and baseband precoders, which reduces the complexity when compared with prior solutions that mostly depend on the joint design of the baseband and RF precoding matrices \cite{ElAyach2014,Alkhateeb2014}, Algorithm \ref{alg:GS_HP} achieves a significant gain over prior solutions, and gives a very close performance to the optimal solution in \eqref{eq:Opt_MI_UP}, as will be shown in \sref{sec:Results}. In fact, for some special cases like the case when $N_\mathrm{S}=N_\mathrm{RF}$, Algorithm \ref{alg:GS_HP} can be proved to provide the optimal baseband and RF precoding design of the problem $\underset{\bF_\mathrm{RF} \in {\cF}_\mathrm{RF}, \left\|\bF_\mathrm{RF} \bF\left[k\right]\right\|^2_F\leq N_\mathrm{RF}}{\max} \frac{1}{K} \sum_{k=1}^K \left\|\tilde{\boldsymbol{\Sigma}}\left[k\right]\tilde{\bV}^*\left[k\right] \bF_\mathrm{RF} \bF\left[k\right]\right\|^2_F$ which has been an important optimization objective for many hybrid precoding papers \cite{ElAyach2014,Alkhateeb2014}.   

\subsection{Total Feedback Overhead}

In this subsection, we summarize the feedback overhead associated with the proposed hybrid precoding strategies in \sref{sec:Codebook} and \sref{sec:Greedy} as illustrated in Table \ref{tab:FB}. 

 In \sref{sec:Codebook}, we develop an algorithm to construct efficient codebooks $\cF_\mathrm{RF}$ for RF precoding \textit{matrices}. Hence, the MS will need to feedback the index of the best RF precoding codeword, i.e., $B_\mathrm{RF}=\log_2\left|\cF_\mathrm{RF}\right|$ bits. For the baseband precoders, we found in \sref{sec:Optimal} that the optimal baseband precoder can be written in terms of the RF precoder and a semi-unitary matrix (the equivalent baseband precoder). In the case when $N_\mathrm{S}=N_\mathrm{RF}$, though, this equivalent baseband precoder takes a unitary structure, and the spectral efficiency is invariant to this equivalent baseband precoder. Hence, no baseband feedback bits are needed in this case. If $N_\mathrm{S}<N_\mathrm{RF}$, $\log_2\left|\cG_\mathrm{BB}\right|$ bit will be needed for each subcarrier.

 In \sref{subsec:Approx_GS}, we developed a greedy hybrid precoding algorithm that requires only a per RF beamforming \textit{vector} codebook $\cF_\mathrm{RF}^\mathrm{v}$. Hence, the index of the best codeword for each RF beamforming vectors will be fed back, i.e., a total of $B_\mathrm{RF}= N_\mathrm{RF} \log_2 \left|\cF_\mathrm{RF}^\mathrm{v}\right|$ bits. The baseband feedback bits are similar to the other scheme.
 \begin{table}[t!]
 	\caption{Total feedback overhead with the proposed limited feedback hybrid precoding strategies}
 	\begin{center}
 		\begin{tabular}{ | c || c | c | c| }
 			\hline
 			\text{Hybrid Precoders Quantization Scheme} & \text{RF Bits $B_\mathrm{RF}$} & \multicolumn{2}{c|}{\text{Baseband Bits $B_\mathrm{BB}$}} \\ \cline{3-4}
 			{} & {} & $N_\mathrm{S}= N_\mathrm{RF}$ & $N_\mathrm{S} < N_\mathrm{RF}$ \\ \hline
 			\makecell{\text{RF matrices quantization with $\cF_\mathrm{RF}$ codebook and baseband matrices} \\ quantization with $\cG_\mathrm{BB}$ codebook as in \sref{sec:Codebook}}& $\log_2\left|\cF_\mathrm{RF}\right|$  &  0  & $K \log_2|\cG_\mathrm{BB}|$ \\ \hline
 			\makecell{\text{RF vectors quantization with $\cF^\mathrm{v}_\mathrm{RF}$ codebook and baseband matrices} \\ quantization with $\cG_\mathrm{BB}$ codebook as in \sref{sec:Codebook}} & $ N_\mathrm{RF} \log_2\left|\cF^\mathrm{v}_\mathrm{RF}\right|$  &  0  & $K \log_2|\cG_\mathrm{BB}|$ \\
 			\hline
 		\end{tabular}
 	\end{center}
 	\label{tab:FB}
 \end{table}
%%%%%%%%%%%%%%%%%%%%%%%%%%%%%%%%%%%%%%%%%%%%%%%%%%%%%%%%%%%%%%%%%%%%%%%%%%%%%%%%%%%%%%%%%%%%%%%%%%%%%%%%%%%%%
\section{Simulation Results} \label{sec:Results}
%%%%%%%%%%%%%%%%%%%%%%%%%%%%%%%%%%%%%%%%%%%%%%%%%%%%%%%%%%%%%%%%%%%%%%%%%%%%%%%%%%%%%%%%%%%%%%%%%%%%%%%%%%%%%
\begin{figure}[t]
	\centering
	\subfigure[center][{}]{
		\includegraphics[width=0.475\columnwidth]{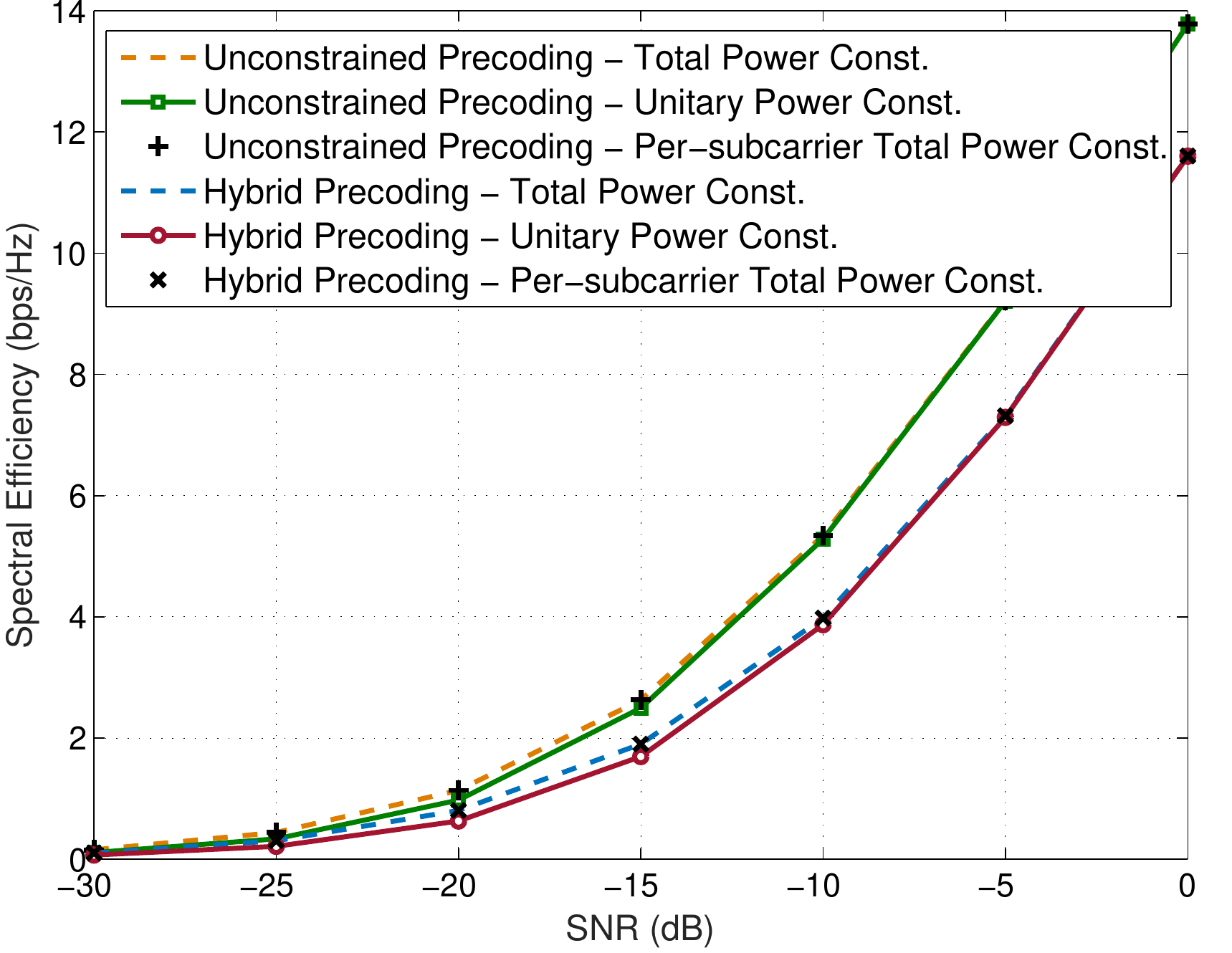}
		\label{fig:Power_SNR}}
	\subfigure[center][{}]{
		\includegraphics[width=0.475\columnwidth]{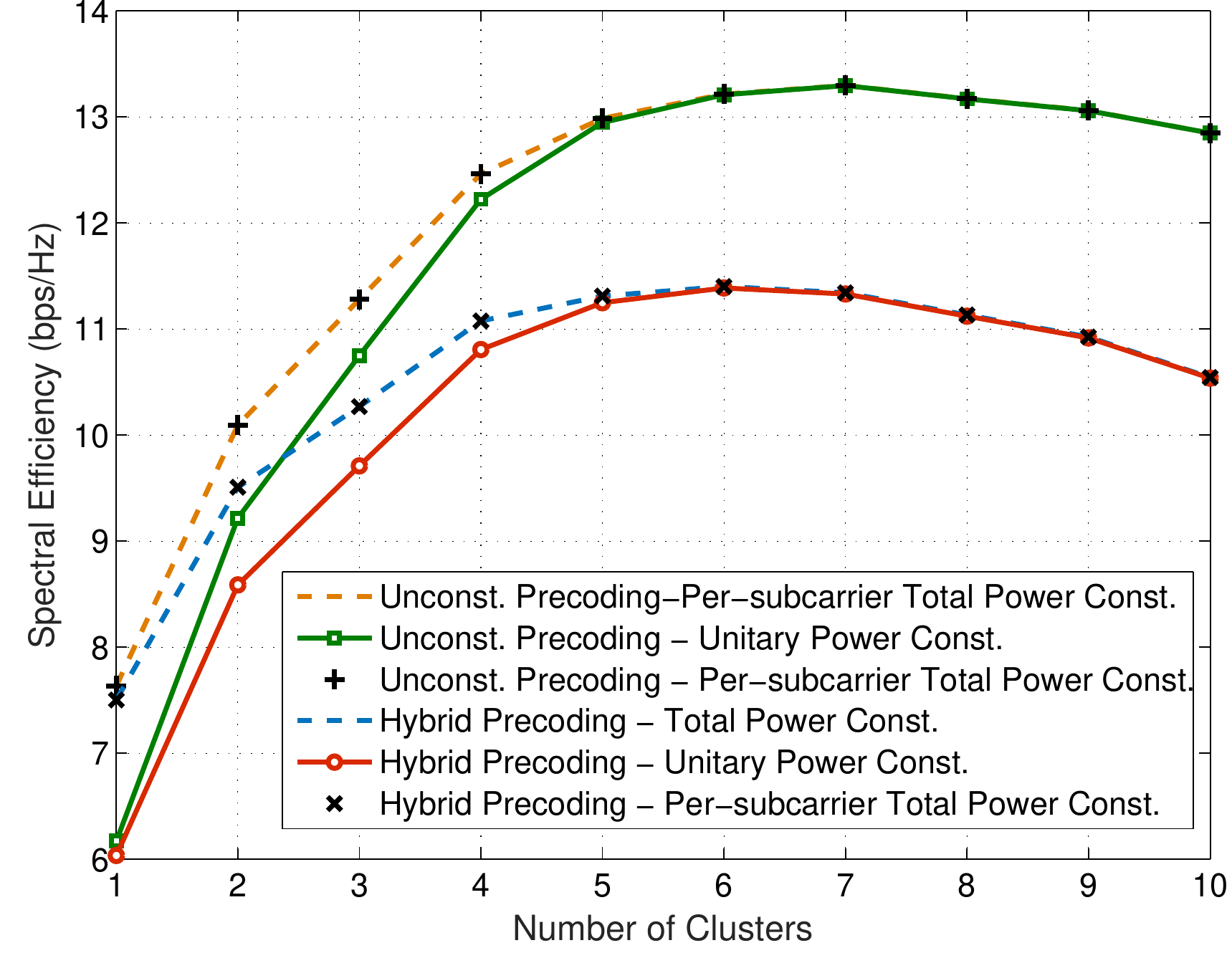}
		\label{fig:Power_Clusters}}
	\caption{The performance of the optimal hybrid precoding design under different power constraints in Proposition \ref{prop:Opt_TP}, Corollary \ref{cor:OPT_TPS}, and Proposition \ref{prop:Opt_UP} versus the SNR in (a) and versus the number of channel clusters with SNR $=0$ dB in (b). The adopted system model has $N_\mathrm{BS}=32$ antennas, $N_\mathrm{MS}=8$ antennas, and $N_\mathrm{S}=N_\mathrm{RF}=3$.}
	\label{fig:Fig1x}
\end{figure}

In this section, we validate our analytical results and evaluate the performance of the proposed codebooks and hybrid precoding designs using numerical simulations. We adopt the wideband mmWave channel as in \eqref{eq:d_Channel}-\eqref{eq:k_Channel}, where a raised-cosine filter is adopted for the pulse shaping function \cite{Schniter2014}, i.e., $p_{\mathrm{rc}}\left(t\right)$ is modeled as
\begin{equation}
p_{\mathrm{rc}}\left(t\right)=\begin{cases} \frac{\pi}{4} \ \sinc\left(\frac{1}{2 \beta}\right), 
& t=\pm \frac{T_\mathrm{s}}{2 \beta} \\ \sinc\left(\frac{t}{T_\mathrm{s}}\right)  \frac{\cos\left(\frac{\pi \beta t}{T_\mathrm{s}}\right)}{1-\left(\frac{2 \beta t}{T_\mathrm{s}}\right)^2},
& \text{otherwise},
\end{cases}
\end{equation} 
with $T_\mathrm{s}$ the sampling time and the roll-off factor $\beta=1$. The number of clusters is assumed to be $L=6$, and the center AoAs/AoDs of the $L$ clusters $\theta_\ell, \phi_\ell$ are assumed to be uniformly distributed in $[0, 2 \pi)$. Each cluster has  $R_\ell=5$ rays  with Laplacian distributed AoAs/AoDs \cite{Forenza2007,ElAyach2014}, and angle spread of $10^\mathrm{o}$. The number of system subcarriers $K$ equals $512$, and the cyclic prefix length is $D=128$, which is similar to 802.11ad \cite{11ad}. The paths delay is uniformly distributed in $[0, D T_\mathrm{s}]$. While the proposed algorithms and codebooks are general for large MIMO channels, we assume in these simulations that both the BS and MS has a ULA with $N_\mathrm{RF}=3$. Hence, $\ba_\mathrm{BS}\left(\phi\right)$ is defined as
\begin{align}
\begin{split}
\ba_\mathrm{BS}\left(\phi\right) =  \frac{1}{\sqrt{N_\mathrm{BS}}} \left[ 1, e^{\j{\frac{2\pi}{\lambda}}d_\mathrm{s}\sin\ \left( \phi \right)}, ..., e^{\j\left(N_\mathrm{BS} -1\right){\frac{2\pi}{\lambda}}d_\mathrm{s}\sin \left( \phi \right)} \right]^\mathrm{T},  \label{eq:steering}
\end{split}
\end{align}
where $\lambda$ is the signal wavelength, and $d_\mathrm{s}$ is the distance between antenna elements with $d_\mathrm{s}=\lambda/2$. The array response vectors at the MS, $\ba_\mathrm{MS}\left(\theta\right)$, can be written in a similar fashion.

%---------------------------------------------------------------------
\subsection{Optimal Hybrid Precoders and Codebook Designs}
%---------------------------------------------------------------------
First, we compare the performance of the optimal hybrid precoders and the fully-digital unconstrained precoders for different power constraints in \figref{fig:Power_SNR}-\figref{fig:Power_Clusters}. For these figures, we adopt the system model in \sref{sec:Model} with $N_\mathrm{BS}=32$ antennas, $N_\mathrm{MS}=8$ antennas, and $N_\mathrm{S}=N_\mathrm{RF}=3$ streams. In \figref{fig:Power_SNR}, the spectral efficiencies of the optimal hybrid precoders under total power constraints, per-subcarrier total power constraints, and unitary power constraints are plotted, and compared with the spectral efficiencies of the SVD unconstrained precoding under the same power constraints. \figref{fig:Power_SNR} shows that the gain of total power constraints over unitary constraints is limited, and decreases with the SNR. The same setup is adopted again in \figref{fig:Power_Clusters} where the optimal hybrid precoders under different power constraints are compared for different numbers of channel clusters, assuming that each cluster contributes with a single ray, and fixing the number of transmitted streams at $N_\mathrm{S}=3$. \figref{fig:Power_Clusters} illustrates that the gain of total power constraints over unitary constraints increases when the channel is very sparse, i.e., when a very small number of clusters exist. This gain, though, is very small if the channel has more than 4-5 clusters.

\begin{figure}[t]
	\centering
	\subfigure[center][{}]{
		\includegraphics[width=0.475\columnwidth]{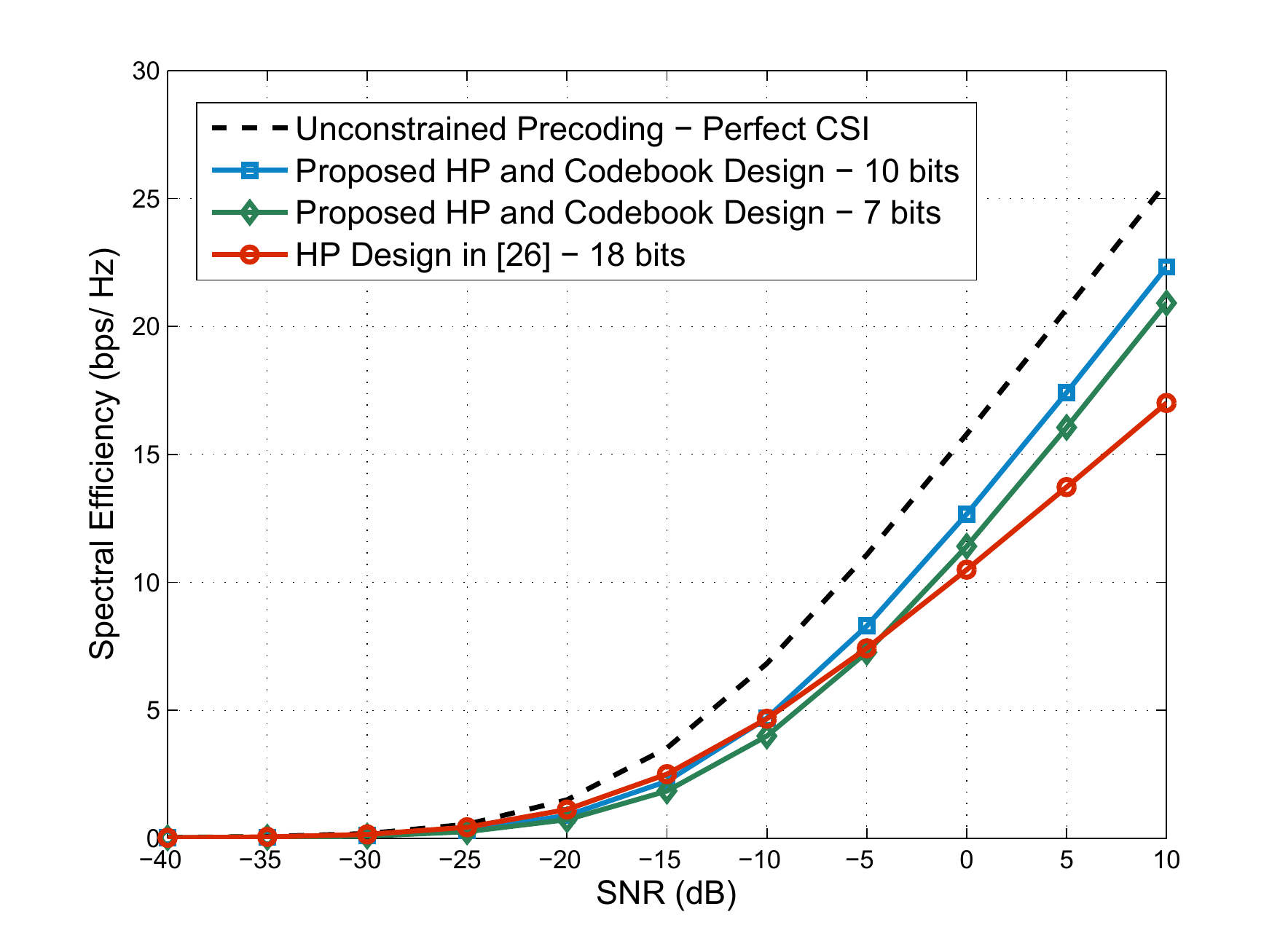}
		\label{fig:Fig2a}}
	\subfigure[center][{}]{
		\includegraphics[width=0.475\columnwidth]{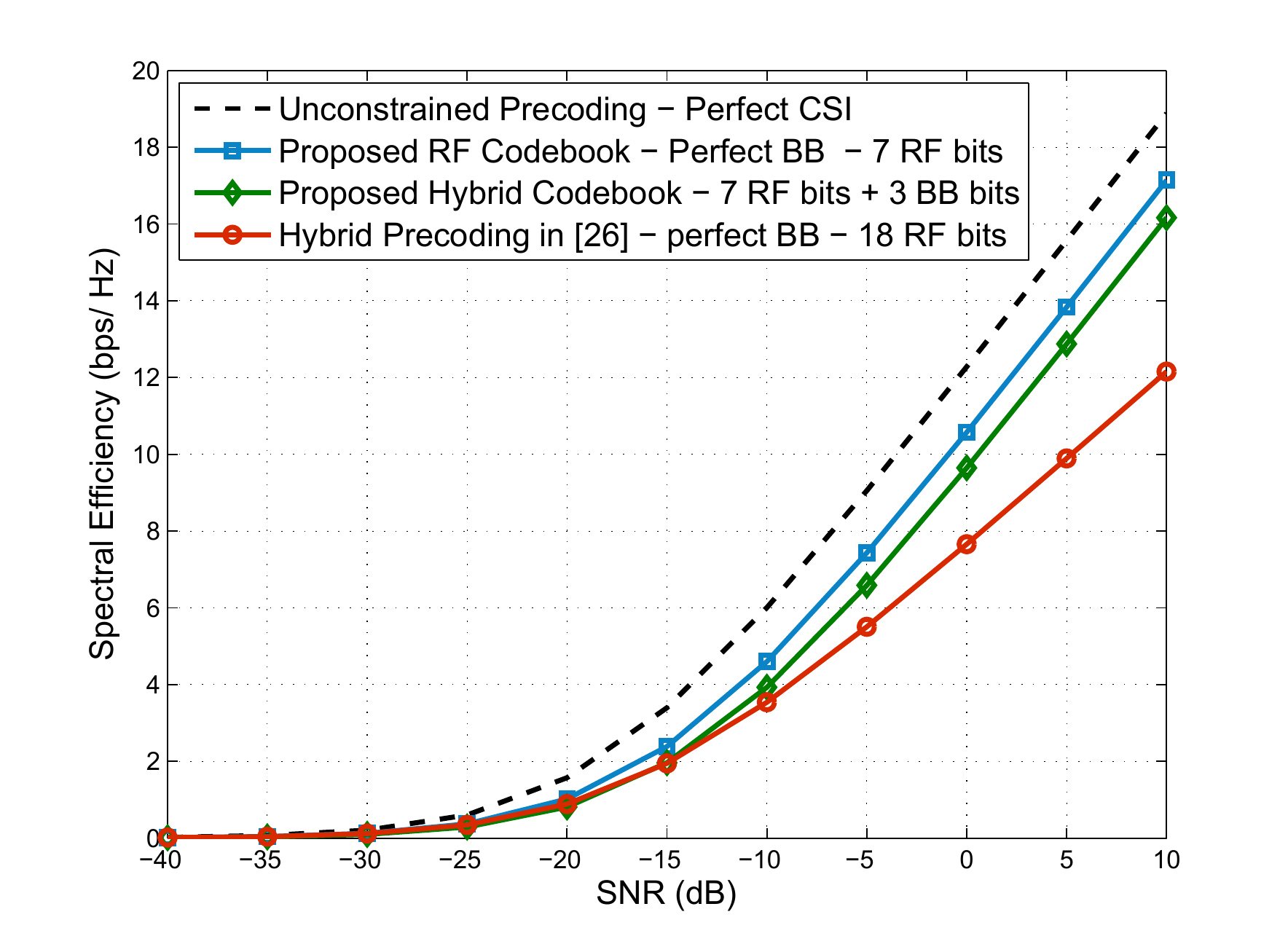}
		\label{fig:Fig2b}}
	\caption{The performance of the proposed hybrid codebook design in Algorithm \ref{alg:RF_CB}, compared with the unconstrained SVD solution, and the prior work in \cite{ElAyach2014}, for the case when $N_\mathrm{S}=N_\mathrm{RF}=3$ in (a), and the case $N_\mathrm{S}=2, N_\mathrm{RF}=3$ in (b). }
	\label{fig:Fig2}
\end{figure}
Next, we evaluate the performance of the proposed hybrid precoding codebooks in \figref{fig:Fig2}, adopting the system model in \figref{fig:Model} with $N_\mathrm{BS}=32$ antennas, and $N_\mathrm{MS}=16$ antennas. In \figref{fig:Fig2a}, the case $N_\mathrm{S}=N_\mathrm{RF}=3$ is considered, the RF codebook is constructed using Algorithm \ref{alg:RF_CB} with different sizes, and the hybrid precoders are designed according to \eqref{eq:Opt_MI_UP}. \figref{fig:Fig2a} shows that the proposed codebook improves the performance compared with the prior work in \cite{ElAyach2014}, even though much smaller numbers of feedback bits are needed, namely $10$ and $7$ bits compared with $18$ bits in the case of beamsteering codebooks in \cite{ElAyach2014,Alkhateeb2014}. In \figref{fig:Fig2b}, the same setup is considered again, but with $N_\mathrm{S}=2$ streams and $N_\mathrm{RF}=3$ RF chains. In this case, the hybrid codebooks are constructed as explained in \sref{subsec:Ns}, with an RF codebook of size $128$ and an equivalent baseband precoders codebook of size $8$. The figure shows that a very  good performance can be also achieved with the designed codebook, despite the relatively small codebook sizes. Further, \figref{fig:Fig2} shows that the proposed limited feedback hybrid precoding codebooks achieve a good slope with the SNR relative to the unconstrained with perfect channel knowledge solution.

\begin{figure}[t]
	\centerline{
		\includegraphics[scale=.6]{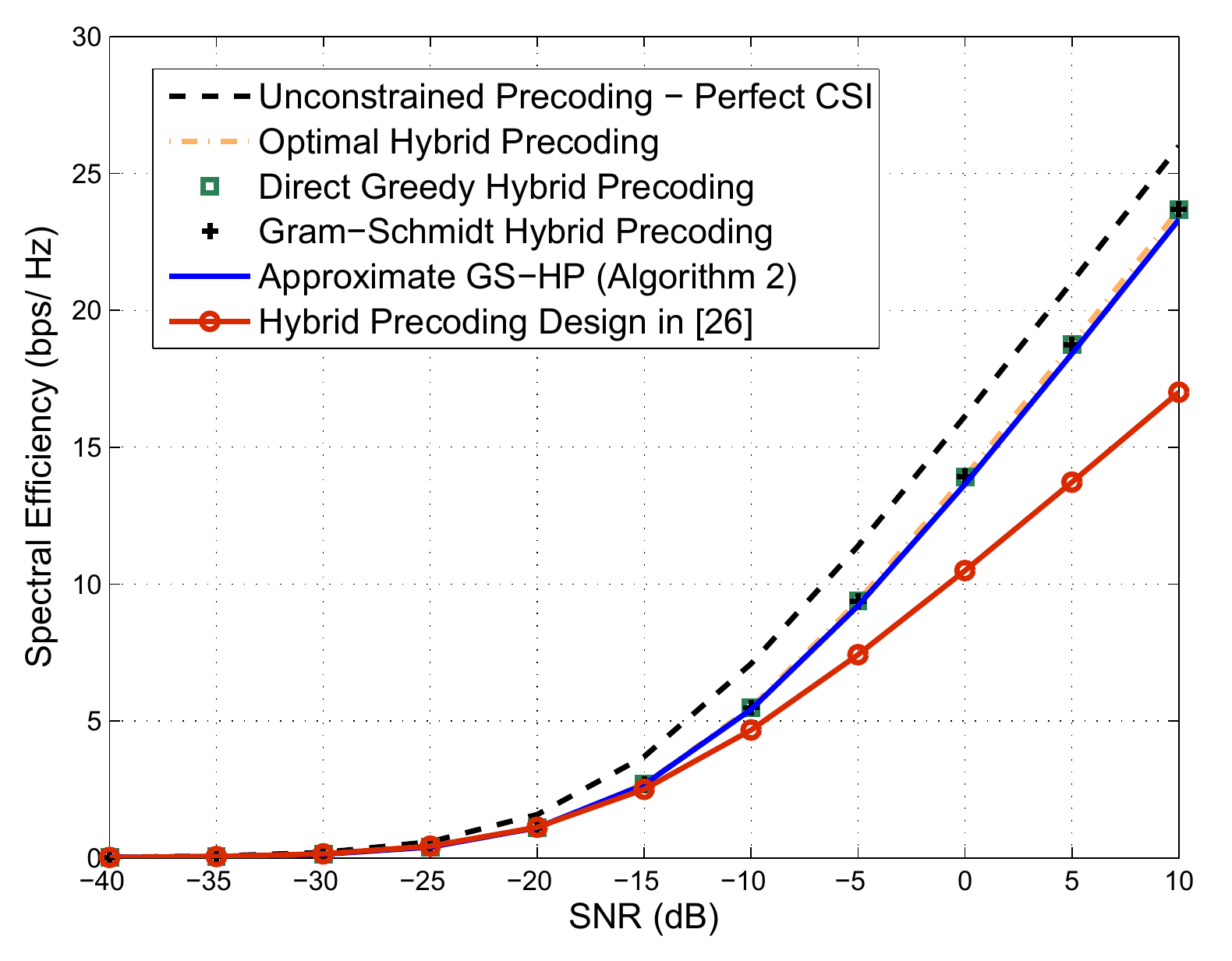}
	}
	\caption{The performance of the approximate Gram-Schmidt hybrid precoding design in Algorithm \ref{alg:GS_HP} compared with the optimal hybrid precoding solution in \eqref{eq:Opt_MI_UP}, the unconstrained SVD solution, and the prior work in \cite{ElAyach2014}. The system has  $N_\mathrm{BS}=32$ antennas, $N_\mathrm{MS}=16$ antennas, and $N_\mathrm{S}=N_\mathrm{RF}=3$.}
	\label{fig:Fig4}
\end{figure}
%---------------------------------------------------------------------
\subsection{Low-Complexity Gram-Schmidt Based Greedy Hybrid Precoding}
%---------------------------------------------------------------------
\begin{figure}[t]
	\centerline{
		\includegraphics[scale=.6]{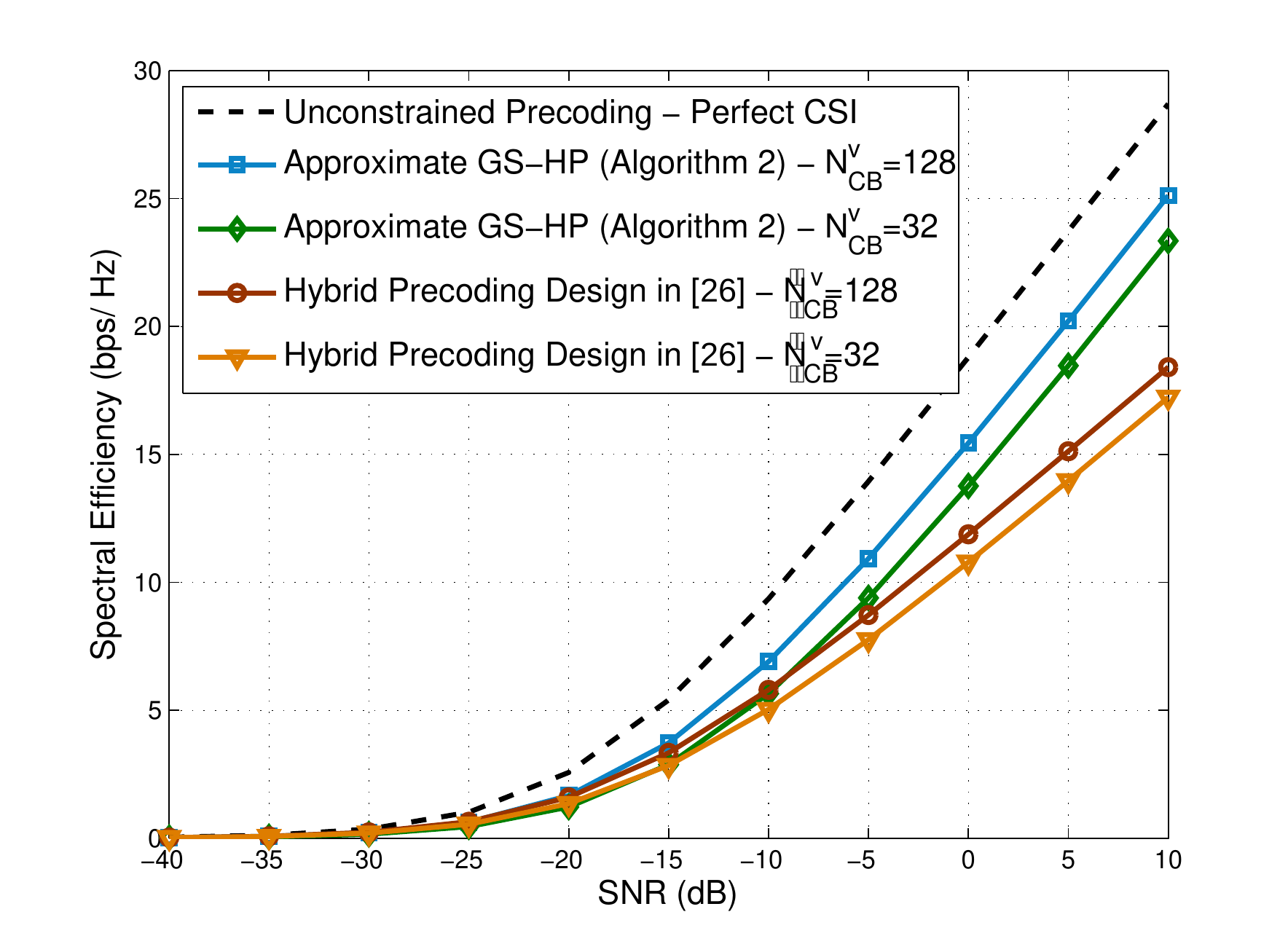}
	}
	\caption{The performance of the approximate Gram-Schmidt hybrid precoding design in Algorithm \ref{alg:GS_HP} for different codebook sizes, compared with the unconstrained SVD solution, and the prior work in \cite{ElAyach2014}. The system has  $N_\mathrm{BS}=64$ antennas, $N_\mathrm{MS}=16$ antennas, and $N_\mathrm{S}=N_\mathrm{RF}=3$.}
	\label{fig:Fig4x}
\end{figure}
In \figref{fig:Fig4} and \figref{fig:Fig4x}, we validate the result in Proposition \ref{prop:GS}, in addition to evaluating the approximate Gram-Schmidt based hybrid precoding algorithm. In \figref{fig:Fig4}, the same setup of \figref{fig:Fig2a} is adopted, and the hybrid precoders are greedily constructed using the direct greedy hybrid precoding algorithm in \eqref{eq:DG_HP}, the Gram-Schmidt hybrid precoding in \eqref{eq:GS_HP}, and the low-complexity approximate Gram-Schmidt hybrid precoding design in Algorithm \ref{alg:GS_HP}. The spectral efficiencies achieved by these greedy algorithms are compared with the optimal hybrid precoding design in \eqref{eq:Opt_MI_UP} where the RF precoders are selected through an exhaustive search over the RF beamforming vectors codebook. The rates are also compared with the prior solution in \cite{ElAyach2014}. For a fair comparison, we assume that each RF beamforming vector is selected from a beamsteering codebook with a size $N_\mathrm{CB}^\mathrm{v}=64$. First, \figref{fig:Fig4} shows that the direct greedy and Gram-Schmidt based hybrid precoding algorithms achieve exactly the same performance which verifies Proposition \ref{prop:GS}. Their performance is also shown to be almost equal to the optimal solution given by \eqref{eq:Opt_MI_UP}. Despite its low-complexity, the developed approximate Gram-Schmidt hybrid precoding design in Algorithm \ref{alg:GS_HP} achieves a very close performance to the optimal solution. We emphasize here that any hybrid precoding design can not perform better that the shown optimal hybrid precoding solution with the considered RF codebook, which confirms the near-optimal result of the proposed algorithm. This is also clear in the considerable gain obtained by the proposed algorithm compared with the prior solution in \cite{ElAyach2014}. Also, it is worth mentioning that the developed hybrid precoding algorithms in this paper can be applied to any large MIMO system (not specifically mmWave systems). The same setup is considered in \figref{fig:Fig4x}, but with $N_\mathrm{BS}=64$ antennas. \figref{fig:Fig4x} illustrates the gain achieved by Algorithm \ref{alg:GS_HP} compared with the designs in \cite{ElAyach2014} for different codebook sizes.

\begin{figure}[t]
	\centerline{
		\includegraphics[scale=.45]{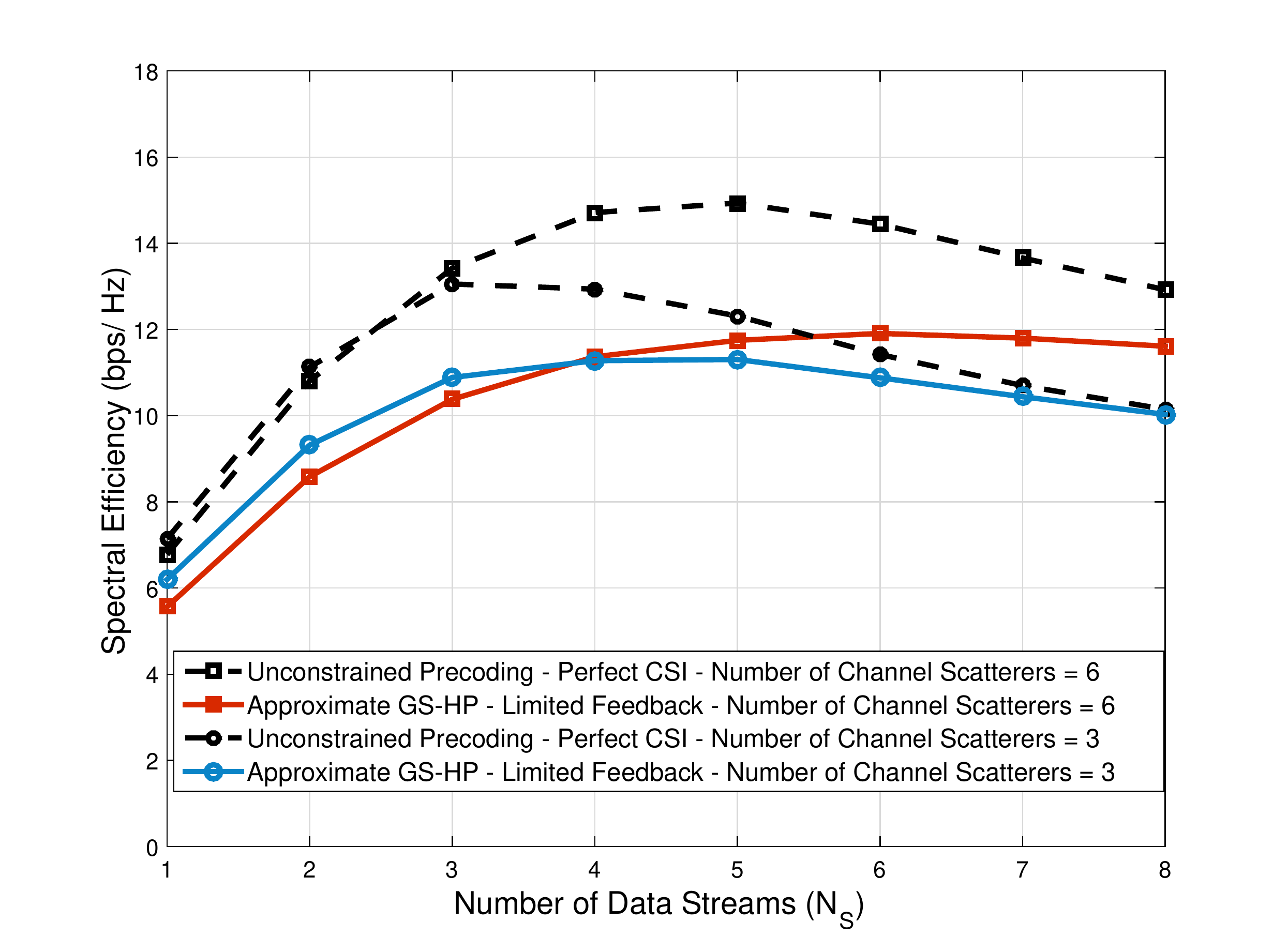}
	}
	\caption{The performance of the approximate Gram-Schmidt hybrid precoding design in Algorithm \ref{alg:GS_HP} compared with the fully-digital SVD solution for different numbers of data streams.  The system has  $N_\mathrm{BS}=32$ antennas, $N_\mathrm{MS}=8$ antennas, and $N_\mathrm{RF}=N_\mathrm{S}$.}
	\label{fig:Fig5}
\end{figure}

In \figref{fig:Fig5}, we evaluate the performance of the proposed approximate Gram-Schmidt hybrid precoding algorithm compared with the digital unconstrained solution for different numbers of transmitted data streams. In this figure, we adopt the same setup of \figref{fig:Fig2a}, but with $N_\mathrm{BS}=32$ antennas and $N_\mathrm{MS}=8$ antennas. Further, each RF beamforming vector is selected from a beamsteering codebook with a size $N_\mathrm{CB}^\mathrm{v}=128$. The number of RF chains are assumed to be equal to the number of data streams. First, \figref{fig:Fig5} shows that the performance of the both the unconstrained precoding and the hybrid precoding increases then decreases again with the number of data streams. This decrease with large numbers of transmitted data streams is a result of the sparse mmWave channels and the equal power allocation among the different streams, which causes some power to be allocated to less important multi-path components. The solution to this problem is what is called multi-mode precoding \cite{Love2005a,Lee2009}. Further, this figure illustrates that the difference between the proposed hybrid precoding algorithm and digital SVD solution is small at both the small number of streams and the large number of streams regimes, which also follows from the sparsity of mmWave channels. 

%-----------------------------------------------------
%-----------------------------------------------------
\subsection{Gain of RF Chains}
%-----------------------------------------------------

\begin{figure}[t]
	\centerline{
		\includegraphics[scale=.58]{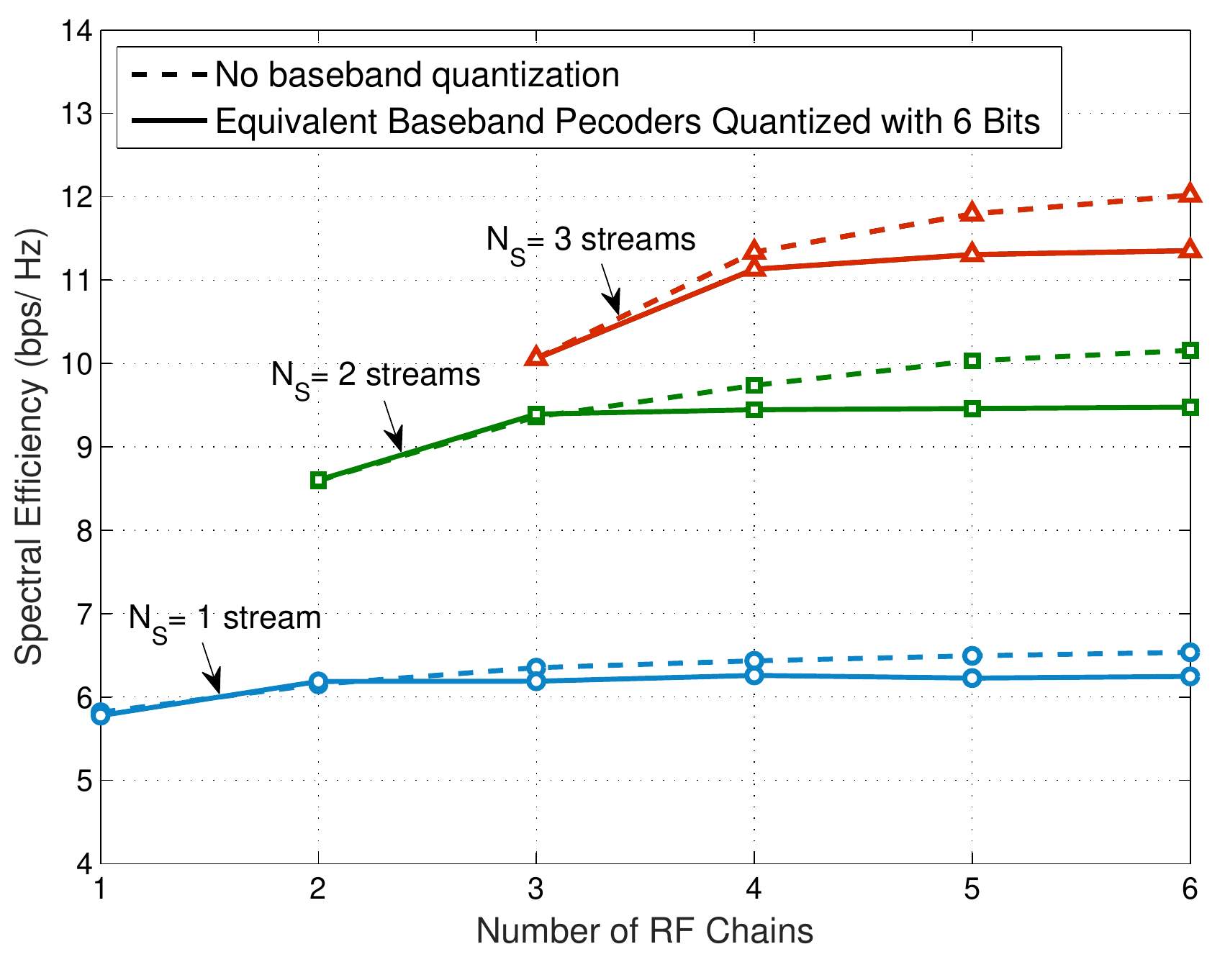}
	}
	\caption{The performance of the approximate Gram-Schmidt hybrid precoding design in Algorithm \ref{alg:GS_HP} compared with the fully-digital SVD solution for different numbers of data streams.  The system has  $N_\mathrm{BS}=32$ antennas, $N_\mathrm{MS}=8$ antennas, and $N_\mathrm{RF}=N_\mathrm{S}$.}
	\label{fig:Fig6}
\end{figure}
 
 In Table \ref{tab:FB}, we summarize the required feedback overhead for the limited feedback operation of  OFDM-based hybrid precoding systems. Table \ref{tab:FB} shows that when the number of transmitted streams equals the number of RF chains, $N_\mathrm{S}=N_\mathrm{RF}$, then only the feedback bits that correspond to the RF precoding codeword need to be fed back to the transmitter. Note that the reason is not that we only need RF beamforming for this case, but because the optimal baseband precoder, as obtained by Proposition \ref{prop:Opt_UP}, can be written as a matrix that depends only on the RF precoder multiplied by a unitary matrix. Table \ref{tab:FB} also shows that the number of feedback bits  scales linearly with the number of subcarriers if $N_\mathrm{S}< N_\mathrm{RF}$. It is, therefore, interesting to evaluate the gain of employing more RF chains than the number of streams, as achieving this gain requires considerable feedback overhead. In \figref{fig:Fig6}, we plot the spectral efficiency achieved with the proposed Gram-Schmidt greedy hybrid precoding versus the number of RF chains for $N_\mathrm{S}=1, 2$, and $3$ streams. The RF beamforming vectors are quantized with $5$ bits and the equivalent baseband precoders are quantized with $6$ bits. This figure shows that the spectral efficiency gain of having more RF chains saturates after a few RF chains. \figref{fig:Fig6} also illustrates that having a number of RF chains $N_\mathrm{RF}= 2 N_\mathrm{S}$ achieves less than $20 \%$ gain, with the cost of much more feedback overhead. The required numbers of feedback bits and the achievable spectral efficiency are also listed in Table \ref{tab:FB_Fig6} for the case $N_\mathrm{S}=2$ streams. Table \ref{tab:FB_Fig6} shows that $3102$ bits are needed to achieve $9.7$ bps/Hz spectral efficiency when $6$ RF chains are employed, while only $10$ bits are enough to obtain $8.6$ bps/Hz when $N_\mathrm{RF}=N_\mathrm{S}=2$. These results indicate that activating a number of RF chains equals to the number of data streams may be a good technique to reduce the feedback overhead and increase the feasibility of limited feedback operation in hybrid precoding based wideband mmWave systems.  

\begin{table}[t!]
 	\caption{Required feedback overhead for the hybrid precoding transmission in \figref{fig:Fig6} with $N_\mathrm{S}=2$ streams}
 \begin{center}
 	\begin{tabular}{ | c | c | c | c | c | c| c | }
 		\hline
 		\text{$\vphantom{\frac{\bX}{\bY}}$Number of RF chains }$N_\mathrm{RF}$ & 2 & 3 & 4 & 5 & 6 \\ \hline
 		\text{$\vphantom{\frac{\bX}{\bY}}$RF feedback bits} $ N_\mathrm{RF} \log_2\left|\cF^\mathrm{v}_\mathrm{RF} \right|$  & 10 \text{bits}& 15 \text{bits}& 20 \text{bits}& 25 \text{bits}& 30 \text{bits} \\ \hline
 		\text{$\vphantom{\frac{\bX}{\bY}}$Baseband feedback bits per subcarrier}  & 0  \text{bits}& 6 \text{bits} & 6 \text{bits}& 6 \text{bits}& 6 \text{bits}\\
 		\hline
 		\text{$\vphantom{\frac{1\bX}{\bY}}$Total feedback bits for $512$ subcarriers}& 10 \text{bits} & 3087 \text{bits} & 3092 \text{bits} & 3097 \text{bits} & 3102 \text{bits} \\ \hline
 		\text{$\vphantom{\frac{\bX}{\bY}}$Spectral efficiency with $N_\mathrm{S}=2$ streams} &8.6 \text{bps/Hz} & 9.4 \text{bps/Hz} & 9.5 \text{bps/Hz} & 9.6 \text{bps/Hz} & 9.7 \text{bps/Hz} \\ \hline
 	\end{tabular}
 \end{center}
 	\label{tab:FB_Fig6}
 \end{table}

%%%%%%%%%%%%%%%%%%%%%%%%%%%%%%%%%%%%%%%%%%%%%%%%%%%%%%%%%%%%%%%%%%%%%%%%%%%%%%%%%%%%%%%%%%%%%%%%%%%%%%%%%%%%%
\section{Conclusion} \label{sec:Conclusion}
%%%%%%%%%%%%%%%%%%%%%%%%%%%%%%%%%%%%%%%%%%%%%%%%%%%%%%%%%%%%%%%%%%%%%%%%%%%%%%%%%%%%%%%%%%%%%%%%%%%%%%%%%%%%%
In this paper, we investigated limited feedback hybrid precoding design for wideband mmWave systems. First, we derived the optimal hybrid precoding design that maximizes the achievable mutual information for any given RF codebook, and showed that the optimal baseband structure can be decomposed into an RF precoder dependent matrix and a unitary matrix. This indicated that when the number of data streams equals to the number of RF chains, only the feedback of the RF precoder index is sufficient to achieve the maximum mutual information. Exploiting the structure of the optimal hybrid precoders, we also showed that the codebook of the equivalent baseband precoders should have a unitary structure. These notes led to efficient hybrid analog/digital precoders codebooks for spatial multiplexing in wideband mmWave systems. Further, we developed a novel greedy hybrid precoding algorithm based on Gram-Schmidt orthogonalization. Thanks to this Gram-Schmidt orthogonalization, we showed that only sequential design of the RF and baseband precoders is required to achieve the same performance of more sophisticated algorithms that requires a joint design of the RF and baseband precoders in each step. Simulation results illustrated that the proposed codebook and precoding algorithms improve over prior work and stay within a small gap from the unconstrained perfect channel knowledge solutions. For future work, it would be interesting to exploit the frequency correlation of mmWave channels to reduce the feedback overhead corresponding to the baseband precoders, which can enhance the feasibility of limited feedback in mmWave systems. Also, another important extension would be to design efficient hybrid precoding codebooks for wideband multi-user millimeter wave systems.

\appendices
\section{}\label{app:prop_OPT_TP}

\begin{proof}[Proof of Proposition \ref{prop:Opt_TP}] Consider the inner problem of \eqref{eq:Opt_CSI_IO}, we first state the following lemma regarding this problem which is inspired by the orthogonalization concept of the nearest matrix problem in \cite{Choi2006}.
\begin{lemma}
Making the change of variables $\bF[k]=\left(\bF_\mathrm{RF}^* \bF_\mathrm{RF}\right)^{-\frac{1}{2}} \bG[k]$, where we call $\bG[k]$ the \textit{equivalent} baseband precoder, the inner problem of \eqref{eq:Opt_CSI_IO} is equivalent to the following problem.
%Lemma 1
\begin{align}
& \underset{ \left\{\bG[k]  \right\}_{ k=1}^K} \max \ \ \frac{1}{K} \sum_{k=1}^{K} \log_2 \left|\bI_{N_\mathrm{S}}+\frac{\rho}{N_\mathrm{S}} \bH[k] \bF_\mathrm{RF} \left(\bF_\mathrm{RF}^* \bF_\mathrm{RF}\right)^{-\frac{1}{2}} \bG[k]  \bG^*[k] \left(\bF_\mathrm{RF}^* \bF_\mathrm{RF}\right)^{-\frac{1}{2}} \bF_\mathrm{RF}^*  \bH^*[k] \right| \nonumber \\
& \hspace{10pt} \text{s.t.} \ \ \ \ \ \sum_{k=1}^K \left\| \bG[k] \right\|_F^2=K N_\mathrm{S}.
\end{align}
\label{lemma1}
\end{lemma}
%Lemma Proof
\begin{proof}
Consider the mapping $\psi: \mathbb{C}^{N_\mathrm{RF} \times N_\mathrm{S}} \rightarrow \mathbb{C}^{N_\mathrm{RF} \times N_\mathrm{S}}$, $\psi\left(\bG[k]\right)=\left(\bF_\mathrm{RF}^* \bF_\mathrm{RF}\right)^{-\frac{1}{2}} \bG[k]=\bF[k]$, and let $\mathcal{S}_{\bF[k]}$ be the set of matrices $\bF[k]$'s defining the domain of the problem \eqref{eq:Opt_CSI_IO}. To prove the equivalence between the problem in Lemma \ref{lemma1}, and the inner problem of \eqref{eq:Opt_CSI_IO}, it is sufficient to prove that $\psi(.)$ is a one-to-one mapping with $\psi\left(\textbf{dom}  \left(\psi\right)\right)$ covering the domain of the problem \eqref{eq:Opt_CSI_IO}, i.e., $\psi\left(\textbf{dom} \left(\psi\right)\right) \supset \mathcal{S}_{\bF[k]}$ \cite{Boyd2004}.

As the matrices $\bF_\mathrm{RF}$ in $\cF_\mathrm{RF}$ are assumed to have linearly independent columns, then the matrix $\bA=\bF_\mathrm{RF}^* \bF_\mathrm{RF}$ is non-singular, and we have $\psi^{-1}\left(\psi\left( \bG[k] \right)\right)= \bA^{\frac{1}{2}}\bA^{-\frac{1}{2}}\bG[k]=\bG[k]$, which proves that $\psi$ is a linear one-to-one mapping. Finally, as the matrix $\bA$ is non-singular, then for any matrix $\bF[k] \in \mathcal{S}_{\bF[k]}$, we have a corresponding $\bG[k] = \bA^{\frac{1}{2}} \bF[k] \in \textbf{dom} \left(\psi\right)$, with $\bA^{- \frac{1}{2}} \bG[k] \in \psi\left(\textbf{dom} \left(\psi\right) \right)$, which implies that $\psi\left(\textbf{dom}\left(\psi\right) \right)\supset \mathcal{S}_{\bF[k]}$, and this completes the proof.
\end{proof}

Now, considering the equivalent optimization problem in Lemma \ref{lemma1}, this problem is a standard MIMO precoder design problem with a mutual information maximization objective, in which the optimal equivalent baseband precoder $\bG^{\star}[k]$ is given by
\begin{equation}
\bG^\star[k]=\left[\overline{\bV}[k]\right]_{:,1:N_\mathrm{S}} \boldsymbol{\Lambda}[k],\ \ k=1, 2, ..., K,
\end{equation}
with $\left[\overline{\bV}[k]\right]_{:,1:N_\mathrm{S}}$ and  $\boldsymbol{\Lambda}[k]$ as defined in Proposition \ref{prop:Opt_TP}. Finally, the one-to-one mapping $\psi$ results in the optimal baseband precoder $\bF^\star[k]$ as defined in equation \eqref{eq:Opt_BB_TP}. It is worth mentioning here that a similar trick to that in Lemma \ref{lemma1} has been concurrently and independently used in \cite{Sohrabi2015} for the heuristic hybrid precoding algorithm proposed in that work, but with no proof on the problems equivalence.
\end{proof}

\section{}\label{app:GS}
\begin{proof}[Proof of Proposition \ref{prop:GS}]
To prove that $\cI_\mathrm{HP}^\mathrm{GS-HP}=\cI_\mathrm{HP}^\mathrm{DG-HP}$, it is sufficient to prove that $\bF_\mathrm{RF}^{(N_\mathrm{RF})}$ of the GS-HP and DG-HP algorithms are equal. To do that, we will show that both the algorithms choose the same RF beamforming vector in each iteration, i.e., $\bF_\mathrm{RF}^{(i)}$ is equal for $i=1,...,N_\mathrm{RF}$. This can be proved using  mathematical induction as follows. At the first iteration, the two algorithms do the exhaustive search over the same codebook $\cF_\mathrm{RF}$, and consequently select the same beamforming vectors. Now, suppose that the two algorithms reach the same RF precoding matrix $\bF_\mathrm{RF}^{(i-1)}$ at iteration $i-1$, we need to prove that they both select the same RF beamforming vector at iteration $i$, i.e., we need to prove that both \eqref{eq:DG_HP} and \eqref{eq:GS_HP} choose beamforming vectors with the same index. To prove that, it is enough to show that the contributions of the $n$th beamforming vector $\bff_n^\mathrm{RF}$ from $\cF_\mathrm{RF}^{\mathrm{v}}$ in \eqref{eq:DG_HP} and \eqref{eq:GS_HP} are equal. i.e., we need to prove that
\begin{align}
\hspace{-10pt}  \log_2 & \left|\bI_{N_\mathrm{S}}+\frac{\rho}{N_\mathrm{S}} \bH[k] \dot{\bF}^{(i,n)}_\mathrm{RF} \left(\dot{\bF}^{(i,n)^*}_\mathrm{RF} \dot{\bF}^{(i,n)}_\mathrm{RF}\right)^{-\frac{1}{2}} \bG^\star[k]  \left(\bG^{\star}[k]\right)^* \left(\dot{\bF}^{(i,n)^*}_\mathrm{RF} \dot{\bF}^{(i,n)}_\mathrm{RF}\right)^{-\frac{1}{2}} \dot{\bF}_\mathrm{RF}^{(i,n)^*}  \bH^*[k] \right| = \nonumber\\
&\hspace{-20pt} \log_2 \left|\bI_{N_\mathrm{S}}+\frac{\rho}{N_\mathrm{S}} \bH[k] \overline{\bF}^{(i,n)}_\mathrm{RF} \left(\overline{\bF}_\mathrm{RF}^{(i,n)^*} \overline{\bF}^{(i,n)}_\mathrm{RF}\right)^{-\frac{1}{2}} \overline{\bG}^\star[k]  \left(\overline{\bG}^{\star}[k]\right)^* \left(\overline{\bF}_\mathrm{RF}^{(i,n)^*} \overline{\bF}^{(i,n)}_\mathrm{RF}\right)^{-\frac{1}{2}} \overline{\bF}_\mathrm{RF}^{(i,n)^*}  \bH^*[k] \right|
\label{eq:Add_MI}
\end{align}
Given the optimal baseband precoder in \eqref{eq:Opt_BB_UP}, and denoting the SVD of $\dot{\bF}_\mathrm{RF}^{(i,n)}$ as  $\dot{\bF}_\mathrm{RF}^{(i,n)}=\dot{\bU}_\mathrm{RF}^{(i,n)} \dot{\boldsymbol{\Sigma}}_\mathrm{RF}^{(i,n)}\dot{\bV}^{(i,n)^*}_\mathrm{RF}$, the LHS of \eqref{eq:Add_MI} can be written as
\begin{align}
\log_2 & \left|\bI_{N_\mathrm{S}}+\frac{\rho}{N_\mathrm{S}} \bH[k] \dot{\bF}^{(i,n)}_\mathrm{RF} \left(\dot{\bF}^{(i,n)^*}_\mathrm{RF} \dot{\bF}^{(i,n)}_\mathrm{RF}\right)^{-\frac{1}{2}} \bG^\star[k]  \left(\bG^{\star}[k]\right)^* \left(\dot{\bF}^{(i,n)^*}_\mathrm{RF} \dot{\bF}^{(i,n)}_\mathrm{RF}\right)^{-\frac{1}{2}} \dot{\bF}^{(i,n)^*}_\mathrm{RF}  \bH^*[k] \right|  \nonumber\\
& \hspace{60pt} = \sum_{\ell=1}^{i} \log_2  \left( 1+\frac{\rho}{N_\mathrm{S}}  \lambda_\ell\left(\bH[k] \dot{\bF}^{(i,n)}_\mathrm{RF} \left(\dot{\bF}^{(i,n)^*}_\mathrm{RF} \dot{\bF}^{(i,n)}_\mathrm{RF}\right)^{-1}  \dot{\bF}^{(i,n)^*}_\mathrm{RF}  \bH^*[k] \right)\right), \label{eq:GS_equal}\\
& \hspace{60pt} = \sum_{\ell=1}^{i} \log_2  \left( 1+\frac{\rho}{N_\mathrm{S}}  \lambda_\ell\left(\bH[k] \dot{\bU}_\mathrm{RF}^{(i,n)} \dot{\bU}_\mathrm{RF}^{(i,n)^*} \bH^*[k] \right)\right).
\end{align}

The RHS of \eqref{eq:Add_MI} can be similarly written, but with $\dot{\bU}_\mathrm{RF}^{(i,n)}$ replaced by $\overline{\bU}_\mathrm{RF}^{(i,n)}$ where $\overline{\bF}_\mathrm{RF}^{(i,n)}=\overline{\bU}_\mathrm{RF}^{(i,n)} \overline{\boldsymbol{\Sigma}}_\mathrm{RF}^{(i,n)} \overline{\bV}^{(i,n)^*}_\mathrm{RF}$. Hence, we  need to prove that $\dot{\bU}_\mathrm{RF}^{(i,n)} \dot{\bU}_\mathrm{RF}^{(i,n)^*}=\overline{\bU}_\mathrm{RF}^{(i,n)} \overline{\bU}_\mathrm{RF}^{(i,n)^*}$. Let $\overline{\bff}_n=\bP^{(i-1)^\perp} \bff_n$ denote the last column of $\overline{\bF}_\mathrm{RF}^{(i,n)}$. As $\overline{\bff}_n$ is a result of successive Gram-Schmidt operations, we can write $\bff_n=\overline{\bff}_n+\bF_\mathrm{RF}^{(i-1)} \boldsymbol{\alpha}_n$, where $\boldsymbol{\alpha}_n$ is a vector results from the Gram-Schmidt process. Consequently,  $\dot{\bF}_\mathrm{RF}^{(i,n)}$ can be written as $\dot{\bF}_\mathrm{RF}^{(i,n)}=\overline{\bF}_\mathrm{RF}^{(i,n)}\bE_C$, where $\bE_C=\left[\begin{array}{cc} \bI & \boldsymbol{\alpha} \\ \boldsymbol{0}^T & 1\end{array}\right]$ is an elementary column operation matrix. Now, we note that $\dot{\bU}_\mathrm{RF}^{(i,n)} \dot{\bU}_\mathrm{RF}^{(i,n)^*}=\dot{\bF}_\mathrm{RF}^{(i,n)} \dot{\bF}_\mathrm{RF}^{(i,n)^{\dagger}}=\overline{\bF}_\mathrm{RF}^{(i,n)} \bE_C \bE_C^{-1} \overline{\bF}_\mathrm{RF}^{(i,n)^{\dagger}}= \overline{\bU}_\mathrm{RF}^{(i,n)} \overline{\bU}_\mathrm{RF}^{(i,n)^*}$, as $\bE_C$ is an $i \times i$ full-rank matrix.
\end{proof}

%------------------------------------------------------------------------------------------------------------
\bibliographystyle{IEEEtran}

\end{document}

%% file: input.tex
\usepackage{amsfonts}
\usepackage{times}
\usepackage{graphicx}
\usepackage{latexsym}
\usepackage{dsfont}
\usepackage{amssymb}
\usepackage{amsmath}
\usepackage{cite}
\usepackage{verbatim}
\usepackage{subfigure}
\newtheorem{theorem}{Theorem}

\newtheorem{algorithm}[theorem]{Algorithm}

\newtheorem{corollary}[theorem]{Corollary}

\newtheorem{lemma}[theorem]{Lemma}

\newtheorem{proposition}[theorem]{Proposition}

\newenvironment{proof}{ \textbf{Proof:} }{ \hfill $\Box$}

\newcommand{\figref}[1]{{Fig.}~\ref{#1}}

\def\bb0{{\mathbb{0}}}

\def\ba{{\mathbf{a}}}
\def\bb{{\mathbf{b}}}

\def\bff{{\mathbf{f}}}

\def\bm{{\mathbf{m}}}
\def\bn{{\mathbf{n}}}

\def\bs{{\mathbf{s}}}

\def\bx{{\mathbf{x}}}
\def\by{{\mathbf{y}}}

\def\b0{{\mathbf{0}}}

\def\bA{{\mathbf{A}}}
\def\bB{{\mathbf{B}}}
\def\bC{{\mathbf{C}}}

\def\bE{{\mathbf{E}}}
\def\bF{{\mathbf{F}}}
\def\bG{{\mathbf{G}}}
\def\bH{{\mathbf{H}}}
\def\bI{{\mathbf{I}}}

\def\bM{{\mathbf{M}}}

\def\bP{{\mathbf{P}}}

\def\bR{{\mathbf{R}}}

\def\bT{{\mathbf{T}}}
\def\bU{{\mathbf{U}}}
\def\bV{{\mathbf{V}}}
\def\bW{{\mathbf{W}}}
\def\bX{{\mathbf{X}}}
\def\bY{{\mathbf{Y}}}

\def\bbE{{\mathbb{E}}}

\def\cA{\mathcal{A}}

\def\cD{\mathcal{D}}

\def\cF{\mathcal{F}}
\def\cG{\mathcal{G}}
\def\cH{\mathcal{H}}
\def\cI{\mathcal{I}}

\def\cN{\mathcal{N}}

\def\cR{\mathcal{R}}

\def\cU{\mathcal{U}}
\def\cV{\mathcal{V}}

\def\sf0{{\mathsf{0}}}

\def\sinc{\mathrm{sinc}}